\documentclass{amsart}

\usepackage{enumerate,amssymb,amsmath,mathrsfs}
\usepackage{amssymb,amsfonts,amsthm,amscd}
\usepackage[mathscr]{eucal}     
\usepackage{graphicx}
\usepackage[arrow, matrix, curve]{xy}
\usepackage{color}
\definecolor{gray}{gray}{.75}
\definecolor{gray2}{gray}{.50}

\newcommand{\dom}{\mathcal{D}}

\newcommand{\N}{\mathbb{N}}
\newcommand{\R}{\mathbb{R}}
\newcommand{\C}{\mathbb{C}}

\newcommand{\A}{\mathrm{\alpha}}
\newcommand{\w}{\mathrm{\omega}}
\newcommand{\la}{\mathrm{\lambda}}

\newtheorem{thm}{Theorem}[section]
\newtheorem{cor}[thm]{Corollary}
\newtheorem{remark}[thm]{Remark}
\newtheorem{prop}[thm]{Proposition}
\newtheorem{lemma}[thm]{Lemma}
\newtheorem{defn}[thm]{Definition}

\numberwithin{equation}{section}

\begin{document}

\title[Functional determinants]
{Functional Determinants for Regular-Singular Laplace-type Operators}
\author{Boris Vertman}
\address{University of Bonn \\ Department of Mathematics \\
Beringstr. 6\\ 53115 Bonn\\ Germany}
\email{vertman@math.uni-bonn.de}


\begin{abstract}
{We discuss a specific class of regular-singular Laplace-type operators with matrix coefficients. Their zeta determinants were studied by K. Kirsten, P. Loya and J. Park on the basis of the Contour integral method, with general boundary conditions at the singularity and Dirichlet boundary conditions at the regular boundary. We complete the arguments of Kirsten, Loya and Park by explicitly verifying that the Contour integral method indeed applies in the regular singular setup. Further we extend the zeta determinant computations to generalized Neumann boundary conditions at the regular boundary and apply our results to Laplacians on a bounded generalized cone with relative ideal boundary conditions.}
\end{abstract}

\maketitle

\pagestyle{myheadings}
\markboth{\textsc{Functional Determinants}}{\textsc{Boris Vertman}} 

\section{Introduction}\
\\[-3mm] The regular-singular Laplace-type operators arise naturally in the context of Riemannian manifolds with conical singularities. The analysis and the geometry of spaces with conical singularities were developped in the classical works of J. Cheeger in [Ch1] and [Ch2]. This setup is modelled by a bounded generalized cone $M=(0,R]\times N, R>0$ over a closed Riemannian manifold $(N,g^N)$ with the Riemannian metric
$$g^M=dx^2 \oplus x^2g^N.$$
In Section \ref{model-section} we discuss closed extensions of regular-singular model operators over $C^{\infty}_0(0,R)$ and self-adjoint extensions of the associated regular-singular model Laplacians. The explicit identification of the relevant domains is used later in the computation of functional determinants. The arguments and results of this section are well-known, however we give a balanced overview and adapt the presentation to applications. 
\\[3mm] The presented calculations go back to J. Br\"{u}ning, R.T. Seeley in [BS] and J. Cheeger in [Ch1] and [Ch2]. For further reference see mainly [W], [BS] and [KLP1], but also [C] and [M].
\\[3mm] In Section \ref{funct-Det} we compute zeta-determinants of regular-singular model Laplacians with matrix coefficients. We consider the case of operators in the limit circle case at the singular boundary $x=0$ and pose there general boundary conditions. This setup has been addressed by K. Kirsten, P. Loya and J. Park in [KLP1] with general boundary conditions at the cone singularity but only with Dirichlet boundary conditions at the regular boundary $x=R$.
\\[3mm] The argumentation of Kirsten, Loya and Park in the articles [KLP1] and [KLP2] is based on the Contour integral method, which gives a specific integral representation of the zeta-function. A priori the Contour integral method need not to apply in the regular-singular setup and is only formally a consequence of the Argument Principle.
\\[3mm] The essential result of Section \ref{funct-Det} is the proof that Contour integral method indeed applies in the regular-singular setup. Our proof is the basis for the integral representation of the zeta-function. Otherwise the Contour integral method would only give information on the pole structure of the zeta-function, but no results on the zeta-determinants. 
\\[3mm] The proof is provided in the setup of generalized Neumann boundary conditions at the cone base, however for Dirichlet boundary conditions the arguments are similar and outlined by the author in the Appendix to [KLP2]. Furthermore we extend the computations of [KLP1] to the setup of generalized Neumann boundary conditions at the cone base, in view of geometric applications below. 
\\[3mm] In Section \ref{section-laplace} we study the relative (ideal) boundary conditions for Laplacians on differential forms, arising from the minimal closed extensions of the exterior derivative. At the cone base $\{R\}\times N$ they are given explicitly by a combination of Dirichlet and generalized Neumann boundary conditions.
\\[3mm] The study of the relative boundary conditions at the cone singularity is interesting on its own. In [BL2], among other issues, the relative extension of the Laplacian on differential forms is shown to coincide with the Friedrich's extension at the cone singularity outside of the "middle degrees". We discuss the relative boundary conditions for Laplace operators on differential forms in any degree and obtain explicit results, relevant for further computations.
\\[3mm] Different sources have analyzed zeta-determinants of Laplace operators over a bounded generalized cone. Under the condition that all self-adjoint extensions of the Laplace operator coincide at the cone singularity, calculations are provided in the joint work of J.S. Dowker and K. Kirsten in [DK]. 
\\[3mm] For the computation of zeta-regularized or so-called "functional" determinants of de Rham Laplacians on a bounded generalized cone it is necessary to note that the Laplacian admits a direct sum decomposition $$\triangle = L \oplus \widetilde{\triangle},$$ which is compatible with the relative boundary conditions, such that $\widetilde{\triangle}$ is the maximal direct sum component, which is essentially self-adjoint at the cone singularity.
\\[3mm] The direct sum component $\widetilde{\triangle}$ is discussed by K. Kirsten and J.S. Dowker in [DK] with general boundary conditions of Dirichlet and Neumann type at the cone base. The other component $L$ is precisely of the type considered in Section \ref{funct-Det}.
\\[3mm] In Section \ref{even} we obtain by a combination of our results with the results by [L], where determinants of a general class of regular-singular Sturm-Liouville operators have been discussed, and [KLP1] an explicit result for the functional determinant of $L$ with relative boundary conditions. In view of [DK] this provides a complete picture of the Laplace operator on a bounded generalized cone.
\\[3mm] {\bf Acknowledgements.} The results of this article were obtained during the author's Ph.D. studies at Bonn University, Germany. The author would like to thank his thesis advisor Prof. Matthias Lesch for his support and useful suggestions. The author is also grateful to Prof. Paul Loya for helpful discussions. The author was supported by the German Research Foundation as a scholar of the Graduiertenkolleg 1269 "Global Structures in Geometry and Analysis".

\section{Regular-Singular Model Operators}\label{model-section}

\subsection{Closed Extensions of Model Operators}\label{model-chapter-weidmann}

Let $D:C^{\infty}_0(0,R)\to C^{\infty}_0(0,R)$, $R>0$ be a differential operator acting on smooth $\C$-valued functions with compact support in $(0,R)$. The standard Hermitian scalar product on $\C$ and the standard measure $dx$ on $\R$ define the natural $L^2$-structure on $C^{\infty}_0(0,R)$:
$$\forall f,g \in C^{\infty}_0(0,R): \ \langle f, g\rangle_{L^2}:=\int_0^R\langle f(x),g(x) \rangle dx.$$
Denote the completion of $C^{\infty}_0(0,R)$ under the $L^2$-scalar product by $L^2(0,R)$. This defines a Hilbert space with the natural $L^2$-Hilbert structure.
\\[3mm] We define the maximal extension $D_{\max}$ of $D$ by $$\dom (D_{\max}):=\{f \in L^2(0,R)|Df \in L^2(0,R)\},\quad D_{\max}f:=Df,$$
where $Df\in L^2(0,R)$ is understood in the distributional sense. The minimal extension $D_{\min}$ of $D$ is defined as the graph-closure of $D$ in $L^2(0,R)$, more precisely:
\begin{align*}
\dom (D_{\min}):=\{f \in L^2(0,R)| \exists (f_n)\subset C^{\infty}_0(0,R)&: \\
f_n \xrightarrow{L^2} f, \quad Df_n \xrightarrow{L^2} Df\} \subseteq \dom (D_{\max})&, \quad D_{\min}f:=Df.
\end{align*}
Analogously we can form the minimal and the maximal extensions of the formal adjoint differential operator $D^t$. Since $C^{\infty}_0(0,R)$ is dense in $L^2(0,R)$, the maximal and the minimal extensions provide densely defined operators in $L^2(0,R)$. In particular we can form their adjoints. The next result provides a relation between the maximal and the minimal extensions of $D, D^t$:
\begin{thm}\label{max-min-theorem}\textup{[W, Section 3]} The maximal extensions $D_{\max},D^t_{\max}$ and the minimal extensions $D_{\min}, D^t_{\min}$ are closed densely defined operators in the Hilbert space $L^2(0,R)$ and are related as follows 
\begin{align}\label{max-min-def}
D_{\max}=(D^t_{\min})^*, \quad D^t_{\max}=(D_{\min})^*.
\end{align}
\end{thm} \ \\
\\[-7mm] The actual discussion in [W, Section 3] is in fact peformed in the setup of symmetric operators. But the arguments there transfer analogously to not necessarily symmetric differential operators. 
\\[3mm] Moreover we introduce the following notation. Let $$C(L^2(0,R))$$ denote the set of all closed extensions $\widetilde{D}$ in $L^2(0,R)$ of differential operators $D$ acting on $C^{\infty}_0(0,R)$, such that $D_{\min} \subseteq \widetilde{D} \subseteq D_{\max}$.
\\[3mm] For the characterization of the von Neumann space $\dom (D_{\max}) / \dom (D_{\min})$ in the setup of symmetric differential operators with real coefficients, the following general concepts in the classical reference [W] become relevant:
\begin{defn} \label{l-p-c-l-c-c}
A symmetric differential operator $D:C^{\infty}_0(0,R)\to C^{\infty}_0(0,R)$ with real coefficients is said to be 
\begin{itemize}
\item in the limit point case $(l.p.c.)$ at $x=0$, if for any $\lambda \in \C$ there is at least one solution $u$ of $(D-\lambda)u=0$ with $u \notin L^2_{\textup{loc}}[0,R)$.
\item in the limit circle case $(l.c.c.)$ at $x=0$, if for any $\lambda \in \C$ all solutions $u$ of $(D-\lambda)u=0$ are such that $u \in L^2_{\textup{loc}}[0,R)$.
\end{itemize}
\end{defn}\ \\
\\[-7mm] Here, $L^2_{\textup{loc}}[0,R)$ denotes elements that are $L^2$-integrable over any closed intervall $I \subset [0,R)$, but not necessarily $L^2$-integrable over $[0,R]$. Furthermore the result [W, Theorem 5.3] implies that if the limit point or the limit circle case holds for one $\lambda \in \C$, then it automatically holds for any complex number. Hence it suffices to check l.p.c or l.c.c. at any fixed $\lambda \in \C$. Similar definition holds at the other boundary $x=R$.
\\[3mm] The central motivation for introducing the notions of limit point and limit circle cases is that it provides a characterization of the von Neumann space $\dom (D_{\max}) / \dom (D_{\min})$ and in particular criteria for uniqueness of closed extensions of $D$ in $C(L^2(0,R))$.
\begin{defn}\label{coincide}
We say that two closed extensions $D_1,D_2$ of a differential operator $D:C^{\infty}_0(0,R) \to C^{\infty}_0(0,R)$ \textup{"coincide" at} $x=0$, if for any cut-off function $\phi \in C^{\infty}[0,R]$ vanishing identically at $x=R$ and being identically one at $x=0$, the following relation holds
\begin{align*}
\phi \mathcal{D}(D_1) = \phi \mathcal{D}(D_2).
\end{align*}
In particular we say that a formally self-adjoint differential operator is \textup{"essentially self-adjoint" at} $x=0$ if all its self-adjoint extensions in $C(L^2(0,R))$ coincide at $x=0$. 
\end{defn} \ \\
\\[-7mm] Similar definition holds at $x=R$. These definitions hold similarly for closed operators in $L^2((0,R),H)$, where $H$ is any Hilbert space. With the introduced notation we obtain as a consequence of arguments behind [W, Theorem 5.4]
\begin{cor}\label{lpc-ess}
Let $D$ be a symmetric differential operator over $C^{\infty}_0(0,R)$ with real coefficients, in the limit point case at $x=0$. Then all closed extensions of $D$ in $C(L^2(0,R))$ coincide at $x=0$ and in particular $D$ is essentially self-adjoint at $x=0$.
\end{cor}

\subsection{First order Regular-Singular Model Operators}
We consider the following regular-singular model operator $$d_p:=\frac{d}{dx}+\frac{p}{x}:C^{\infty}_0(0,R)\to C^{\infty}_0(0,R), \quad p \in \R.$$
Any element of the maximal domain $\dom (d_{p,\max})$ is square-integrable with its weak derivative in $L^2_{loc}(0,R]$, due to regularity of the coefficients of $d_p$ at $x=R$. So we have (compare [W, Theorem 3.2])$$\dom(d_{p,\max})\subset H^1_{loc}(0,R].$$ Consequently elements of the maximal domain $\dom(d_{p,\max})$ are continuous at any $x \in (0,R]$. Further we derive by solving the inhomogeneous differential equation $d_pf=g\in L^2(0,R)$ via the variation of constants method (the solution to the homogeneous equation $d_pu=0$ is simply $u(x)=c\cdot x^{-p}$), that elements of the maximal domain $f \in \dom (d_{p,\max})$ are of the following form 
\begin{align}\label{max-form}
f(x)=c\cdot x^{-p}-x^{-p}\cdot \int_x^Ry^p(d_pf)(y)dy.
\end{align}
We now analyze the expression above in order to determine the asymptotic behaviour at $x=0$ of elements in the maximal domain of $d_p$ for different values of $p\in \R$.
\begin{prop}\label{d-p-max}
Let $O(\sqrt{x})$ and $O(\sqrt{x |\log (x)|})$ refer to the asymptotic behaviour as $x\to 0$. Then the maximal domain of $d_p$ is characterized explicitly as follows:
\begin{enumerate}
\item For $p < -1/2$ we have
$$\dom(d_{p,\max})=\{f \in H^1_{loc}(0,R]|f(x)=O(\sqrt{x}), d_pf \in L^2(0,R)\}.$$ 
\item For $p = -1/2$ we have
\begin{align*}
\dom(d_{p,\max})=\{f \in H^1_{loc}(0,R]|f(x)=O(\sqrt{x|\log x|}), d_pf \in L^2(0,R)\}.
\end{align*} 
\item For $p \in (-1/2;1/2)$ we have
\begin{align*}
\dom(d_{p,\max})=\{f \in H^1_{loc}(0,R]|f(x)=c_fx^{-p}+O(\sqrt{x}), d_pf \in L^2(0,R)\},
\end{align*}
where the constants $c_f$ depend only on $f$.
\item For $p \geq 1/2$ we have $$\dom(d_{p,\max})=\{f \in H^1_{loc}(0,R]|f(x)=O(\sqrt{x}), d_pf \in L^2(0,R)\}.$$
\end{enumerate}
\end{prop}
\begin{proof} 
Due to similarity of arguments we prove the first statement only, in order to avoid repetition. Let $p < -1/2$ and consider any $f \in \dom(d_{p,\max})$. By \eqref{max-form} this element can be expressed by $$f(x)=c\cdot x^{-p}-x^{-p}\cdot \int_x^Ry^pg(y)dy,$$ where $g= d_pf$. By the Cauchy-Schwarz inequality we obtain for the second term in the expression 
\begin{align*}
\left| x^{-p}\int^R_xy^pg(y)dy\right|& \leq x^{-p}\sqrt{\int^R_xy^{2p}dy}\cdot \sqrt{\int^R_xg^2} \leq \\ \leq c \cdot x^{-p}&\sqrt{x^{2p+1}-R^{2p+1}}\|g\|_{L^2}=c \cdot \sqrt{x}\sqrt{1-R^{2p+1}x^{-2p-1}}\|g\|_{L^2},
\end{align*}
where $c=1/\sqrt{-2p-1}$. Since $(-2p-1) > 0$ we obtain for the asymptotics as $x \to 0$ $$ x^{-p}\int^R_xy^pg(y)dy = O(\sqrt{x}).$$ Observe further that for $p<-1/2$ we also have $x^{-p}=O(\sqrt{x})$. This shows the inclusion $\subseteq$ in the statement. To see the converse inclusion observe $$\{f \in H^1_{loc}(0,R]|f(x)=O(\sqrt{x}), \ \textup{as} \ x \to 0\} \subset L^2(0,R).$$ This proves the statement.
\end{proof}\ \\
\\[-7mm] In order to analyze the minimal closed extension $d_{p,\min}$ of $d_p$, we need to derive an identity relating $d_p$ to its formal adjoint $d^t_p$, the so-called Lagrange identity. With the notation of Proposition \ref{d-p-max} we obtain the following result.
\begin{lemma}\label{lagrange-identity}(Lagrange-Identity) For any $f \in \dom (d_{p,\max})$ and $g \in \dom (d^t_{p,\max})$ 
\begin{align*}
\left<d_{p}f,g\right>-\left<f, d^t_{p}g\right> &= f(R)\overline{g(R)} -c_f \overline{c_g}, \ \textup{for} \ |p|<1/2, \\
\left<d_{p}f,g\right>-\left<f, d^t_{p}g\right> &= f(R)\overline{g(R)}, \ \textup{for} \ |p|\geq 1/2.
\end{align*}
\end{lemma} 
\begin{proof}
$$\left<d_{p}f,g\right>-\left<f, d^t_{p}g\right>= f(R)\overline{g(R)} -f(x)\cdot \overline{g(x)}|_{x\to 0}.$$
Applying Proposition \ref{d-p-max} to $f\in \dom (d_{p,\max})$ and $g\in \dom (d^t_{p,\max})=\dom (d_{-p,\max})$ we obtain:
\begin{align*}
f(x)\cdot \overline{g(x)}|_{x\to 0}&=c_f \overline{c_g}, \ \textup{for} \ |p|<1/2, \\
f(x)\cdot \overline{g(x)}|_{x\to 0}&=0, \ \textup{for} \ |p|\geq 1/2, 
\end{align*}
This proves the statement of the lemma.
\end{proof}
\begin{prop}\label{lagrange}
\begin{align*}
\dom(d_{p,\min})&=\{f \in \dom (d_{p,\max})| c_f=0, \, f(R)=0\}, \ \textup{for} \ |p|<1/2, \\
\dom(d_{p,\min})&=\{f \in \dom (d_{p,\max})| f(R)=0\}, \ \textup{for} \ |p|\geq1/2, 
\end{align*}
where the coefficient $c_f$ refers to the notation in Proposition \ref{d-p-max} (iii).
\end{prop}
\begin{proof}
Fix some $f \in \dom(d_{p,\min})$. Then for any $g \in \dom(d^t_{p,\max})$ we obtain using $d_{p,\min}=(d^t_{p,\max})^*$ (see Theorem \ref{max-min-theorem}) the following relation: $$\left<d_{p,\min}f,g\right>-\left<f, d^t_{p,\max}g\right>=0.$$ Together with the Lagrange identity, established in Lemma \ref{lagrange-identity} we find 
\begin{align}\label{eqn}
f(R)\overline{g(R)} -c_f \overline{c_g}&=0, \ \textup{for} \ |p|<1/2, \\
f(R)\overline{g(R)} &=0, \ \textup{for} \ |p|\geq 1/2.
\end{align}
Let now $|p|<1/2$. Then for any $c,b \in \C$ there exists $g \in \dom (d^t_{p,\max})$ such that $c_g=c$ and $g(R)=b$. By arbitrariness of $c,b \in \C$ we conclude from \eqref{eqn} $$c_f=0, \quad f(R)=0.$$ For $|p|\geq 1/2$ similar arguments hold, so we get $f(R)=0$. This proves the inclusion $\subseteq$ in the statements. For the converse inclusion consider some $f \in \dom (d_{p,\max})$ with $c_f=0$ (for $|p|<1/2$) and $f(R)=0$. Now for any $g \in \dom(d^t_{p,\max})$ we infer from Lemma \ref{lagrange-identity} 
$$\left<d_{p,\max}f,g\right>-\left<f, d^t_{p,\max}g\right>=0.$$ 
Thus $f$ is automatically an element of $\dom((d_{p,\max}^t)^*)=\dom(d_{p,\min})$. This proves the converse inclusion.
\end{proof} \ \\
\\[-7mm] Now by a direct comparison of the results in Propositions \ref{d-p-max} and \ref{lagrange} we obtain the following corollary.
\begin{cor}\label{ess-s.a.} \ \\[-4mm]
\begin{enumerate}
\item For $|p| \geq 1/2$ the closed extensions $d_{p,\min}$ and $d_{p,\max}$ coincide at $x=0$.
\item For $|p|< 1/2$ the asymptotics of elements in $\dom (d_{p,\max})$ differs from the asymptotics of elements in $\dom (d_{p,\min})$ by presence of $u(x):=c\cdot x^{-p}$, solving $d_pu=0$.
\end{enumerate}
\end{cor}
\begin{remark}
The calculations and results of this subsection are the one-dimensional analogue of the discussion in [BS]. In particular, the result of Corollary \ref{ess-s.a.} corresponds to [BS, Lemma 3.2].
\end{remark}

\subsection{Self-adjoint extensions of Model Laplacians}\label{model-operator}
Let the model Laplacian be the following differential operator 
$$\triangle:=-\frac{d^2}{dx^2}+\frac{\lambda}{x^2}:C^{\infty}_0(0,R)\to C^{\infty}_0(0,R),$$ where we assume $\lambda \geq -1/4$. Put $$p:=\sqrt{\lambda +\frac{1}{4}}-\frac{1}{2}\geq -\frac{1}{2}.$$
In this notation we find $$\triangle=d_p^td_p=:\triangle_p.$$
Recall that the maximal domain $\dom (\triangle_{p,\max})$ is defined as follows
$$\dom (\triangle_{p,\max})=\{f\in L^2(0,R)|\triangle_p f\in L^2(0,R)\}.$$
Hence any element of the maximal domain is square-integrable with its second and thus also its first weak-derivative in $L^2_{loc}(0,R]$. So we have (compare [W, Theorem 3.2]) 
\begin{align}\label{H-2-loc}
\dom (\triangle_{p,\max})\subset H^2_{loc}(0,R].
\end{align}
We determine the maximal domain $\dom (\triangle_{p,\max})$ explicitly, see also the classical calculations provided in [KLP1].
\begin{prop}\label{laplacian-maximal}
Let $O(x^{3/2})$ and $O(x^{1/2})$ refer to the asymptotic behaviour as $x\to 0$. Then the maximal domain $\dom (\triangle_{p,\max})$ of the Laplace operator $\triangle_p$ is characterized explicitly as follows:
\begin{align*}
\textup{(i)} \ &For \ p=-1/2 \ we \ have \\
&\dom (\triangle_{p,\max})=\{f \in H^2_{loc}(0,R]| f(x)=c_1(f) \cdot \sqrt{x} + c_2(f)\cdot \sqrt{x}\log (x) + \nonumber \\ & \qquad + \widetilde{f}(x), \
\widetilde{f}(x)=O(x^{3/2}), \ \widetilde{f}'(x)=O(x^{1/2}), \ \triangle_p\widetilde{f}(x)\in L^2(0,R)\}. \nonumber \\[3mm]
\textup{(ii)} \  & For \ |p|< 1/2 \ we \ have \\
&\dom (\triangle_{p,\max})= \{f \in H^2_{loc}(0,R]| , f(x)=c_1(f) \cdot x^{p+1} + c_2(f) \cdot x^{-p} + \nonumber \\ & \qquad + \widetilde{f}(x), \
\widetilde{f}(x)=O(x^{3/2}), \ \widetilde{f}'(x)=O(x^{1/2}), \ \triangle_p\widetilde{f}(x)\in L^2(0,R)\}. \nonumber \\[3mm]
\textup{(iii)} \ & For \ p \geq 1/2 \ we \ have \\
&\dom (\triangle_{p,\max})=\{f \in H^2_{loc}(0,R]| f(x)=O(x^{3/2}), \nonumber \\ &\qquad f'(x)=O(x^{1/2}), \ \triangle_pf(x)\in L^2(0,R)\}. \nonumber
\end{align*}
The constants $c_{1}(f),c_2(f)$ depend only on the function $f$.
\end{prop}
\begin{proof}
Consider any $f \in \dom (\triangle_{p,\max}), p\geq -1/2$ and note that $\triangle_p=d_p^td_p=-d_{-p}d_p$. Hence we have the inhomogeneous differential equation $d_{-p}(d_pf)=-g$ with $g\equiv\triangle_p f \in L^2(0,R)$.
\\[3mm] Analogous situation has been considered in Proposition \ref{d-p-max}. Repeating the arguments there we obtain
\begin{align}\label{dpf}
(d_pf)(x)=c\cdot x^{p}+A(x),
\end{align} 
where $A(x)= O(\sqrt{x}), x\to 0$ for $p\neq 1/2$ and $A(x)=O(\sqrt{x |\log (x)|}), x\to 0$ for $p=1/2$. Note by \eqref{H-2-loc} that the functions $d_pf(x)$ and $A(x)$ are continuous at any $x\in (0,R]$. Applying the variation of constants method to the differential equation in \eqref{dpf} we obtain
\begin{align}
f(x)&=\textup{const}\cdot x^{-p}-x^{-p}\int_x^Ry^p(d_pf)(y)
dy = \nonumber \\ &= \textup{const}\cdot x^{-p}-\textup{const} \cdot x^{-p}\int_x^Ry^{2p}dy-x^{-p}\int_x^Ry^pA(y)dy = \nonumber \\
&= \textup{const}\cdot x^{-p}-\textup{const} \cdot x^{-p}\int_x^Ry^{2p}dy+x^{-p}\int_0^xy^pA(y)dy, \label{f-f-tilde}
\end{align}
where "const" denotes any constant depending only on $f$ and the last equality follows from the fact that $id^p\cdot A\in L^1(0,R)$ for $p\geq -1/2$, due to the asymptotics of $A(y)$ as $y\to 0$. Put
\begin{align*}
\widetilde{f}(x)= x^{-p}\int_0^xy^pA(y)dy.
\end{align*}
Using the asymptotic behaviour of $A(y)$ as $y\to 0$, we derive $\widetilde{f}(x)=O(x^{3/2})$ and $\widetilde{f}'(x)=O(x^{1/2})$ as $x\to 0$. Evaluating now explicitly the second integral in \eqref{f-f-tilde} for different values of $p\geq -1/2$ and noting for $p\geq 1/2$ the facts that $x^{p+1}=O(x^{3/2})$ and $id^{-p}\notin L^2(0,R)$, we prove the inclusion $\subseteq$ in the statement on the domain relations. 
\\[3mm] For the converse inclusion observe that any $f \in H^2_{loc}(0,R]$ with the asymptotic behaviour as $x\to 0$:
\begin{align*}
f(x)=c_1(f) \cdot \sqrt{x} + c_2(f)\cdot \sqrt{x}\log (x) + O(x^{3/2}), \ &\textup{for} \ p=-1/2, \\
f(x)=c_1(f)\cdot x^{p+1} + c_2(f) \cdot x^{-p} + O(x^{3/2}), \ &\textup{for} \ |p|<1/2, \\
f(x)=O(x^{3/2}), \ &\textup{for} \ p\geq 1/2,
\end{align*}
is square integrable, $f\in L^2(0,R)$. It remains to observe why $\triangle_p f\in L^2(0,R)$ for any $f$ in the domains on the right hand side of the statement. This becomes clear, once we note that the additional terms in the asymptotics of $f$ other than $\widetilde{f}(x)$ are solutions to $\triangle_pu=0$.
\end{proof}\ \\
\\[-7mm] In order to analyze the minimal closed extension $\triangle_{p,\min}$ we need to derive the Lagrange identity for $\triangle_p$, see also [KLP2, Section 3.1].
\begin{lemma}\label{lapl-lagr} [Lagrange-identity] 
For any $f,g \in \dom (\triangle_{p,\max})$ the following identities hold.
\begin{enumerate}
\item If $p=-1/2$, then we have in the notation of Proposition \ref{laplacian-maximal}
\begin{align*}
\langle f, \triangle_p g\rangle_{L^2}-\langle \triangle_p f, g\rangle_{L^2} =\\
= [c_1(f) \overline{c_2(g)}-c_2(f) \overline{c_1(g)}] + [f'(R)\overline{g(R)}-f(R) \overline{g'(R)}].
\end{align*}
\item If $|p|<1/2$, then we have in the notation of Proposition \ref{laplacian-maximal}
\begin{align*}
\langle f, \triangle_p g\rangle_{L^2}-\langle \triangle_p f, g\rangle_{L^2} =\\
= -(2p+1)[c_1(f) \overline{c_2(g)}-c_2(f) \overline{c_1(g)}] + [f'(R)\overline{g(R)}-f(R) \overline{g'(R)}].
\end{align*}
\item If $p\geq 1/2$, then we have
\begin{align*}
\langle f, \triangle_p g\rangle_{L^2}-\langle \triangle_p f, g\rangle_{L^2}
=[f'(R)\overline{g(R)}-f(R) \overline{g'(R)}].
\end{align*}
\end{enumerate}
\end{lemma}
\begin{proof}
Let $f,g \in \dom (\triangle_{p,\max})$ be any two elements of the maximal domain of $\triangle_p$. We compute:
\begin{align*}
\langle f, \triangle_p g\rangle_{L^2}-\langle \triangle_p f, g\rangle_{L^2} =\\
=\lim_{\epsilon \to 0} \int_{\epsilon}^R [f(x)\overline{\triangle_pg(x)} -\triangle_pf(x)\overline{g(x)}]dx= \\
=\lim_{\epsilon \to 0} \int_{\epsilon}^R \frac{d}{dx}[-f(x)\overline{g'(x)} + f'(x)\overline{g(x)}]dx=\\
=\lim_{\epsilon \to 0}[f(\epsilon)\overline{g'(\epsilon)}-f'(\epsilon)\overline{g(\epsilon)}] + [f'(R)\overline{g(R)}-f(R)\overline{g'(R)}].
\end{align*}
Now the statement follows by inserting the asymptotics at $x=0$ of $f,g \in \dom (\triangle_{p,\max})$ into the first summand of the expression above. 
\end{proof}

\begin{prop}\label{laplacian-minimal} The minimal domain of the model Laplacian $\triangle_p$ is given explicitly in the notation of Proposition \ref{laplacian-maximal} as follows
\begin{align*}
\dom (\triangle_{p,\min})&= \\ =\{f \in \dom &(\triangle_{p,\max})|c_1(f)=c_2(f)=0, f(R)=f'(R)=0\}, \ p\in [-1/2,1/2), \\
\dom &(\triangle_{p,\min})=\{f \in \dom (\triangle_{p,\max})| f(R)=f'(R)=0\}, \ p \geq 1/2.
\end{align*}
\end{prop}
\begin{proof}
Fix some $f \in \dom (\triangle_{p,\min})$. Then for any $g \in \dom (\triangle_{p,\max})$ we obtain with $\triangle_{p,\min}=\triangle_{p,\max}^*$ (see Theorem \ref{max-min-theorem}) the following relation
$$\langle f, \triangle_p g\rangle_{L^2}-\langle \triangle_p f, g\rangle_{L^2}=0.$$
Together with the Lagrange-identity, established in Lemma \ref{lapl-lagr}, and the fact that for $p\in [-1/2,1/2)$ and any arbitrary $c_1,c_2,b_1,b_2 \in \C$ there exists $g \in \dom (\triangle_{p,\max})$ such that $$c_1(g)=c_1, \ c_2(g)=c_2, \ g(R)=b_1, \ g'(R)=b_2,$$
we conclude for $f\in \dom (\triangle_{p,\min}), p\in [-1/2,1/2)$
\begin{align}\label{helping2}
c_1(f)=c_2(f)=0, f(R)=f'(R)=0.
\end{align} 
Analogous arguments for $f\in \dom (\triangle_{p,\min}), p\geq 1/2$ show $f(R)=f'(R)=0$. This proves the inclusion $\subseteq$ in the statement. For the converse inclusion consider any $f\in \dom (\triangle_{p,
\max})$, satisfying \eqref{helping2}, where the condition $c_1(f)=c_2(f)=0$ is imposed only for $p\in [-1/2,1/2)$. Now we obtain from the Lagrange-identity in Lemma \ref{lapl-lagr}
$$\forall g \in \dom (\triangle_{p,\max}): \ \langle f, \triangle_p g\rangle_{L^2}-\langle \triangle_p f, g \rangle_{L^2}=0.$$
Hence $f$ is automatically an element of $\dom (\triangle_{p.\max}^*)=\dom (\triangle_{p,\min})$. This proves the converse inclusion. 
\end{proof}

\begin{cor}\label{ess.s.a}
\begin{enumerate}
\item For $\lambda \geq 3/4$, equivalently for $p=\sqrt{\lambda+1/4}-1/2\geq 1/2$, the model Laplacian $\triangle_p$ is in the limit point case at $x=0$ and the closed extensions $\triangle_{p,\max}$ and $\triangle_{p,\min}$ coincide at $x=0$. In particular $\triangle_p$ is essentially self-adjoint at $x=0$.
\item For $\lambda \in [-1/4, 3/4)$, equivalently for $p=\sqrt{\lambda+1/4}-1/2\in [-1/2,1/2)$, the model Laplacian $\triangle_p$ is in the limit circle case at $x=0$ and the asymptotics at zero of the elements in $\dom (\triangle_{p,\max})$ differ from the asymptotics at zero of elements in $\dom (\triangle_{p,\min})$ by presence of fundamental solutions to $\triangle_pu=0$.
\end{enumerate}
\end{cor}
\begin{proof}
On the one hand statements on the coincidence or the difference of maximal and minimal domains at $x=0$ follow from a direct comparison of the results of Propositions \ref{laplacian-maximal} and \ref{laplacian-minimal}. It remains then to verify the limit point and the limit circle statements. They follow by definition from the study of the fundamental solutions $u_1,u_2:(0,R)\to \R$ of $\triangle_pu=0$:
\begin{align}
&\textup{For} \ p=-1/2 \quad u_1(x)=\sqrt{x}, \ u_2(x)=\sqrt{x}\log(x), \label{fundamental1}\\
&\textup{For} \ p>-1/2 \quad u_1(x)=x^{p+1}, \ u_2(x)=x^{-p}.\label{fundamental2}
\end{align} 
\end{proof}\ \\
\\[-7mm] Next, since the model Laplacian $\triangle_p$ is shown to be essentially self-adjoint at $x=0$ for $p\geq 1/2$, we are interested in the self-adjoint extensions of $\triangle_p$ for $p \in [-1/2,1/2)$, since only there the boundary conditions at $x=0$ are not redundant. In this subsection we determine for these values of $p$ the two geometrically meaningsful extensions of the model Laplacian $-$ the D-extension and N-extension:
\begin{align*}
\triangle_p^{D}:=(d_{p,\min})^*(d_{p,\min})=d^t_{p,\max}d_{p,\min}, \\ \triangle_p^{N}:=(d_{p,\max})^*(d_{p,\max})=d^t_{p,\min}d_{p,\max}.
\end{align*}
\begin{cor}\label{ND-extension}
For $|p| < 1/2$ we have in the notation of Proposition \ref{laplacian-maximal}
\begin{align}
\dom(\triangle_p^{D})=\{f \in \dom(\triangle_{p,\max})| c_2(f)=0,\ f(R)=0\}, \label{D-extension} \\ \dom(\triangle_p^{N})=\{f \in \dom(\triangle_{p,\max})| c_1(f)=0,\ d_pf(R)=0\}. \label{N-extension}
\end{align}
\end{cor}
\begin{proof}
Let us consider the D$-$extension first. By definition $\dom (\triangle_p^D)\subset \dom (d_{p,\min})$ and thus by Proposition \ref{lagrange} we have for any $f\in \dom (\triangle_p^D)$ $$f(x)=O(\sqrt{x})\quad \textup{and} \quad f(R)=0.$$ Since $x^{-p} \neq O(\sqrt{x}),x\to 0$ for $|p|< 1/2$ we find in the notation of Proposition \ref{laplacian-maximal} that the constant $c_2(f)$ must be zero for $f\in \dom (\triangle^D_p)$. This proves the inclusion $\subseteq$ in the first statement. 
\\[3mm] For the converse inclusion consider $f \in \dom(\triangle_{p,\max})$ with $c_2(f)=0, \ f(R)=0$. By Proposition \ref{lagrange} we find $f \in \dom (d_{p,\min})$. Now with $f \in \dom (\triangle_{p,\max})$ we obtain $d_{p,\min}f\in \dom(d^t_{p,\max})$ and hence $$f\in \dom (\triangle_p^D).$$ This proves the converse inclusion of the first statement.
\\[3mm] For the second statement consider any $f \in \dom (\triangle_p^N)=\dom(d^t_{p,\min}d_{p,\max})$. There exists some $g \in \dom(d^t_{p,\min})$ such that $d_pf=g$ with the general solution of this differential equation obtained by the variation of constants method. $$f(x)=c\cdot x^{-p}-x^{-p}\int_0^xy^pg(y)dy.$$ Since $g \in \dom(d^t_{p,\min})$ and thus in particular $g(x)=O(\sqrt{x})$, we find via the Cauchy-inequality that the second summand in the solution above behaves as $O(x^{3/2})$ for $x\to 0$. Hence $c_1(f)=0$ in the notation of Proposition \ref{laplacian-maximal}. Further $d_pf \in \dom(d^t_{p,\min})$ and thus $$d_pf(R)=0.$$ This proves the inclusion $\subseteq$ in the second statement. 
\\[3mm] For the converse inclusion consider an element $f \in \dom(\triangle_{p,\max})$ with $f(x)=c_2(f)\cdot x^{-p} + \widetilde{f}(x)$, where $\widetilde{f}(x)=O(x^{3/2}), \widetilde{f}'(x)=O(\sqrt{x})$, as $x\to 0$, and $d_pf(R)=0$. The inclusion $f \in \dom(d_{p, \max})$ is then clear by Proposition \ref{d-p-max}. Now by Proposition \ref{lagrange} we have $d_pf=d_p\widetilde{f}\in \dom (d_{-p,\min})$ due to asymptotics of $\widetilde{f}(x)$ as $x\to 0$ and $d_pf(R)=0$. 
\end{proof}
\begin{cor}\label{DN-extension-1/2}
For $p =-1/2$ we have  in the notation of Proposition \ref{laplacian-maximal} $$\dom(\triangle_p^{D}) =\{f \in \dom(\triangle_{p,\max})| c_2(f)=0,\ f(R)=0\}.$$ $$\dom(\triangle_p^{N}) =\{f \in \dom(\triangle_{p,\max})| c_2(f)=0,\ d_pf(R)=0\}.$$
\end{cor}
\begin{proof}The first statement is proved by similar arguments as in Corollary \ref{D-extension}. The Corollary \ref{ess-s.a.} asserts the equality of the D-extension and the N-extension at $x=0$ in the sense of Definition \ref{coincide}. This determines the asymptotic behaviour of $f \in \dom(\triangle_p^{N})$ as $x\to 0$. For the boundary conditions of $\triangle_p^{N}$ at $x=R$ simply observe that for any $f\in \dom (\triangle_p^{N})$ one has in particular $d_pf \in \dom(d_{-p,\min})$ and hence $d_pf(R)=0$. 
\end{proof}
\begin{remark}
The naming "D-extension" and "N-extension" coincides with the convention chosen in [LMP, Section 2.3]. However the motivation for this naming is given here by the type of the boundary conditions at the regular end $x=R$. In fact $\dom(\triangle_p^{D})$ has Dirichlet boundary conditions at $x=R$ and $\dom(\triangle_p^{D})$ $-$ generalized Neumann boundary conditions at $x=R$.
\end{remark}\ \\
\\[-7mm] So far we considered the self-adjoint extensions of $\triangle_p=d_p^td_p$ with $p:=\sqrt{\lambda+1/4}-1/2\in [-1/2,1/2)$. However for $r=-p-1$ we have $$d_r^td_r=d_p^td_p=-\frac{d^2}{dx^2}+\frac{\lambda}{x^2}, \ \textup{since}\ r(r+1)=p(p+1)=\lambda. $$
Hence for completeness it remains identify the D- and the N-extensions for $d_r^td_r, r=p-1\in (-3/2,-1/2]$ as well. Note however that for $p=-1/2$ we get $r=p=-1/2$ and the D-, N-extensions are as established before. It remains to consider $r \in (-3/2, -1/2)$.
\begin{cor}\label{DN-extension}
Let $p \in (-1/2;-1/2)$. Put $r=-p-1\in (-3/2,-1/2)$. Then we have in the notation of Proposition \ref{laplacian-maximal} $$\dom(\triangle_r^{D}) =\{f \in \dom(\triangle_{p,\max})| c_2(f)=0,\ f(R)=0\},$$ $$\dom(\triangle_r^{N}) =\{f \in \dom(\triangle_{p,\max})| c_2(f)=0,\ d_
rf(R)=0\}.$$
\end{cor}
\begin{proof}The first statement is proved by similar arguments as in Corollary \ref{ND-extension}. Further, Corollary \ref{ess-s.a.} implies equality of the D-extension and the N-extension at $x=0$ in the sense of Definition \ref{coincide}. This determines the asymptotic behaviour of $f \in \dom(\triangle_r^{N})$ at $x=0$. For the boundary conditions of $\triangle_r^{N}$ at $x=R$ simply observe that for any $f\in \dom (\triangle_r^{N})$ one has in particular $d_rf \in \dom(d_{-r,\min})$ and hence $d_rf(R)=0$. 
\end{proof}

\section{Functional Determinants for Regular-Singular Sturm-Liouville Operators}\label{funct-Det}
\subsection{Self-adjoint Realizations}
We consider the following model setup. Let the operator $L$ be the following regular-singular Sturm-Liouville operator 
$$L= -\frac{d^2}{dx^2}+ \frac{1}{x^2}A:C^{\infty}_0((0,R),\C^q)\to C^{\infty}_0((0,R),\C^q),$$
where for any fixed $q\in \N$, $C^{\infty}_0((0,R),\C^q)$ denotes the space of smooth $\C^q$-valued functions with compact support in $(0,R)$. Let the tangential operator $A$ be a symmetric $q\times q$ matrix and choose on $\C^q$ an orthonormal basis of $A$-eigenvectors. Then we can write:
$$L=\bigoplus\limits_{\lambda \in \textup{Spec}(A)}-\frac{d^2}{dx^2}+\frac{\lambda}{x^2}.$$
Following [KLP1, KLP2] we need a classification of boundary conditions at $x=0$ for self-adjoint realizations of $L$. In view of Corollary \ref{ess.s.a} we restrict to the case $$\textup{Spec}(A) \subset [-1/4,3/4),$$
so that $L$ is a finite direct sum of model Laplace operators in the limit circle case at $x=0$ and $x=R$. In this case boundary conditions must be posed at both boundary components.
\\[3mm] Fix a counting on Spec$(A)$ as follows
$$-\frac{1}{4}=\lambda_1=\cdots =\lambda_{q_0}< \lambda_{q_0+1}\leq \cdots \leq \lambda_{q=q_0+q_1}<\frac{3}{4}.$$ 
Denote by $E_{l}$ the $\lambda_l$-eigenspace of $A$. We count the eigenvalues of $A$ with their multiplicities, so $E_{l}$ is understood to be one-dimensional with $E_{l}=\langle e_{l}\rangle$. Over $C^{\infty}_0((0,R),E_{l}), l=1, ..., q$ the differential operator $L$ reduces to a model Laplace operator $$-\frac{d^2}{dx^2}+\frac{\lambda_l}{x^2}: C^{\infty}_0(0,R) \to C^{\infty}_0(0,R).$$ 
Consider the maximal closed extension $L_{\max}$ of the differential operator $L$. Any $\phi_l \in \dom (L_{\max})\cap L^2((0,R),E_{l}),l=1,...,q$ is given by $f_l\cdot e_{l}$ where $e_{l}$ is the generator of the one-dimensional eigenspace $E_{l}$ and by definition (see also \eqref{H-2-loc}) $$f_l \in \dom \left( -\frac{d^2}{dx^2} +\frac{\lambda_l}{x^2}\right)_{\max}\subset H^2_{loc}(0,R].$$ We identify any $\phi_l$ with its scalar part $f_l$ and observe by Proposition \ref{laplacian-maximal} that $f_l(x)$ has the following asymptotics at $x=0$: \begin{align}
&c_l\sqrt{x}+c_{q+l}\sqrt{x}\log x+ O(x^{3/2}), \ \textup{as} \ l=1,...,q_0.  \label{asymptotics1}\\
&c_l x^{\nu_l+\frac{1}{2}} + c_{q+l} x^{-\nu_l+\frac{1}{2}}+ O(x^{3/2}), \ \textup{as} \  l=q_0+1,...,q.\label{asymptotics2} \\
&\textup{with } \nu_l:=\sqrt{\lambda_l+\frac{1}{4}}. \nonumber
\end{align}
A general element $\phi \in \dom (L_{\max})$ decomposes into a direct sum of such $\phi_l, l=1,...,q$, each of them of the asymptotics above. This defines a vector for any $\phi \in \dom (L_{\max})$ $$\vec{\phi}:=(c_1,...c_{2q})^T \in \C^{2q}.$$ 
Consider now any $\phi , \psi \in \dom (L_{\max})$ and the associated vectors $\vec{\phi}, \vec{\psi}$. Each of the components $\phi_l, \psi_l, l=1, .., q$ lies in the maximal domain of the corresponding model Laplace operator and thus is continuous over $(0,R]$ and differentiable over $(0,R)$ with the derivatives $\phi'_l, \psi'_l$ extending continuously to $x=R$. We impose boundary conditions at $x=R$ as follows 
$$\phi '(R)+ \A \phi (R) =0, \quad \psi '(R)+ \A \psi (R) =0,$$
where these equations are to be read componentwise and $\A \in \R$. 
\\[3mm] In [KLP2, Section 3] the classical results on self-adjoint extensions are reviewed and based on the Lagrange identity for $L$ (similar to Lemma \ref{lapl-lagr}) the boundary conditions at $x=0$ for the self-adjoint extensions of $L$ are classified in terms of Lagrangian subspaces.
\\[3mm] As a consequence of [KLP2, Corollary 3.4, (4.2)] the self-adjoint realizations of $L$ with fixed generalized Neumann boundary conditions at $x=R$ are characterized as follows:
\begin{align}\label{L-s.a.}
\dom (\mathcal{L})=\{\phi \in \dom (L_{\max})| \phi'(R)+\A \phi(R)=0, \ (\mathcal{A} \ \mathcal{B})\vec{\phi}=0\},
\end{align}
where the matrices $\mathcal{A}, \mathcal{B}\in \C^{q\times q}$ are fixed according to the conditions of [KLP2, Corollary 3.4], i.e. $(\mathcal{A} \ \mathcal{B})\in \C^{q\times 2q}$ is of full rank $q$ and $\mathcal{A}'\mathcal{B}^*$ is self-adjoint, where $\mathcal{A}'$ is the matrix $\mathcal{A}$ with the first $q_0$ columns multiplied by $(-1)$.
\begin{remark}
Note that in general $\mathcal{L}$ does not decompose into a finite sum of one-dimensional boundary value problems, since the matrices $(\mathcal{A}, \mathcal{B})$ need not to be diagonal. Therefore the computations below do not reduce to a one-dimensional discussion.
\end{remark}

\subsection{Functional Determinants}\label{neumann-determinants}
In this subsection we continue with the analysis of the self-adjoint realization $\dom (\mathcal{L})$, fixed in \eqref{L-s.a.}. Our aim is to construct explicitly the analytic continuation of the associated zeta-function to $s=0$ and to compute the zeta-regularized determinant of $\mathcal{L}$. 
\\[3mm] We follow the ideas of [KLP1, KLP2], where however only Dirichlet boundary conditions at the cone base $x=R$ have been considered. We extend their approach to generalized Neumann boundary conditions at the cone base, in order to apply the calculations to the relative self-adjoint extension of the Laplace Operator on a bounded generalized cone.
\\[3mm] Furthermore we put the arguments on a thorough footing by proving applicability of the Contour integral method in the regular-singular setup.
\\[3mm] We introduce the following $q\times q$ matrices in terms of Bessel functions of first and second kind:
\begin{align*}
J^{\pm}(\mu):=\!\left(\begin{array}{l} \left(\!\kappa J_{\pm 0}(\mu R)+\mu\sqrt{R}J'_{\pm 0}(\mu R)\right)\cdot Id_{q_0} \qquad 0 \\ \ 0 \quad \textup{diag}\left[2^{\pm \nu_l}\Gamma(1\pm\nu_l)\mu^{\mp\nu_l}\left(\kappa J_{\pm\nu_l}(\mu R) +\mu\sqrt{R}J'_{\pm\nu_l}(\mu R)\right)\right]\end{array}\!\right)\!,
\end{align*}
where the diagonal block matrix in the right low corner has entries for $l=q_0+1,...,q$. Further we have introduced new constants $$\kappa := \frac{1}{2\sqrt{R}}+\alpha \sqrt{R}, \quad \nu_l:=\sqrt{\lambda_l+\frac{1}{4}}, \ l=q_0+1,...,q,$$ to simplify notation. Moreover the function $J_{-0}(\mu R)$ is defined as follows:
$$J_{-0}(\mu x):=\frac{\pi}{2}Y_0(\mu x) - (\log\mu -\log 2 + \gamma)J_0(\mu x)$$ with $\gamma$ being the Euler constant and where we fix for the upcoming discussion the branch of logarithm in $\C\backslash \R^+$ with $0 \leq Im \log (z) < 2\pi$. With this notation we can now formulate the implicit eigenvalue equation for $\mathcal{L}$.
\begin{prop}\label{implicit-eigenvalue-equation}
$\mu^2$ is an eigenvalue of $\mathcal{L}$ if and only if the following equation is satisfied \begin{align*}
F(\mu):=\textup{det}\left(\begin{array}{cc}\mathcal{A} & \mathcal{B} \\ J^{+}(\mu) & J^{-}(\mu)\end{array}\right)\stackrel{!}{=}0.
\end{align*} 
\end{prop}
\begin{proof} Any $\mu^2$-eigenvector $\phi$ of $\mathcal{L}$ is given by a direct sum of scalar functions $\phi_l, l=1,..,q$, which are in the $\lambda_l$-eigenspace of $A$ for any fixed $x\in (0,R]$. Each $\phi_l$ arises as a solution to the Bessel equation $$-\phi_l''+\frac{\lambda_l}{x^2}\phi_l=\mu^2\phi_l.$$ Putting $\nu_l:=\sqrt{\lambda_l+1/4}$ we can rewrite the equation as follows:
\begin{align}\label{bessel-equation-nu}
-\phi_l''+\frac{1}{x^2}\left(\nu_l-\frac{1}{4}\right)\phi_l=\mu^2\phi_l.
\end{align}
The general solution to this Bessel equation is given in terms of $J_*$ and $Y_*$, Bessel functions of first and second kind, respectively.
\\[3mm] For $l=1,..,q_0$ we have $\lambda_l=-1/4$ and hence $\nu_l=0$. In this case the Bessel equation \eqref{bessel-equation-nu} has two linearly independent solutions, $\sqrt{x}J_0(\mu x)$ and $\sqrt{x}Y_0(\mu x)$. Following [KLP2, Section 4.2] we write for the general solution
\begin{align}
c_l\sqrt{x}J_0(\mu x)+c_{q+l}\sqrt{x}J_{-0}(\mu x), \ \textup{as} \ l=1,...,q_0, \label{bessel-solution1}
\end{align}
where $c_l, c_{q+l}$ are constants and $J_{-0}(\mu x)=\frac{\pi}{2}Y_0(\mu x) - (\log\mu -\log 2 + \gamma)J_0(\mu x)$ with $\gamma$ being the Euler constant. Note from [AS, p.360] with $H_k=1+1/2+...+1/k$ and $z \in \C$:
$$\frac{\pi}{2}Y_0(z)=(\log z -\log 2 + \gamma)J_0(z)-\sum\limits_{k=1}^{\infty}\frac{H_k(-z^2/4)^k}{(k!)^2},$$
with $H_k=1+1/2+..+1/k$. Thus by definition we obtain $$J_{-0}(\mu x)=\log x \cdot J_0(\mu x)-\sum\limits_{k=1}^{\infty}\frac{H_k(-(\mu x)^2/4)^k}{(k!)^2}. $$
For $l=q_0+1,..,q$ we have $\lambda_l\in (-1/4, 3/4)$ and hence $\nu_l\in (0,1)$, in particular $\nu_l$ is non-integer. In this case the Bessel equation \eqref{bessel-equation-nu} has two linearly independent solutions $\sqrt{x}J_{\nu_l}(\mu x)$ and $\sqrt{x}J_{-\nu_l}(\mu x)$. Following [KLP2, Section 4.2] we write for the general solution
\begin{align}
\nonumber\\
c_l2^{\nu_l}\Gamma(1+\nu_l)\mu^{-\nu_l}\sqrt{x}J_{\nu_l}(\mu x) + c_{q+l} 2^{-\nu_l} \Gamma(1-\nu_l) \mu^{\nu_l} \sqrt{x}J_{-\nu_l}(\mu x),\label{bessel-solution2} \\ \ \textup{as} \ l=q_0+1,...,q. \nonumber
\end{align}
Now we deduce from the standard series representation of Bessel functions [AS, p. 360] the following asymptotic behaviour as $x\to 0$:
\begin{align}
\label{ass1} \sqrt{x}J_0(\mu x)=\sqrt{x}+\sqrt{x}O(x^2), \\
\label{ass2} \sqrt{x}J_{-0}(\mu x)=\sqrt{x}\log x + \sqrt{x}\log x \cdot O(x^2)+ \sqrt{x}O(x^2), \\
\label{ass3} 2^{\pm \nu_l}\Gamma(1\pm\nu_l)\mu^{\mp\nu_l}\sqrt{x} J_{\pm\nu_l}(\mu x)=x^{\pm \nu_l+1/2} + x^{\pm \nu_l+1/2}O(x^2),
\end{align}
where $O(x^2)$ is given by power-series in $(x\mu)^2$ with no constant term. Hence the asymptotic behaviour at $x=0$ of the general solutions \eqref{bessel-solution1} and \eqref{bessel-solution2} corresponds to the asymptotics \eqref{asymptotics1} and \eqref{asymptotics2}, respectively. Organizing the constants $c_l,c_{q+l}, l=1,..,q$ into a vector $\vec{\phi}=(c_1,...,c_{2q})$, we obtain $$(\mathcal{A},\mathcal{B})\vec{\phi}=0, $$ since by assumption, $\phi \in \dom (\mathcal{L})$. We now evaluate the generalized Neumann boundary conditions at the regular boundary $x=R$. 
\begin{align*}
\phi'_{l}(R)+\alpha\phi_{l}(R)=0 \ \Rightarrow \  c_l\cdot \{(\frac{1}{2\sqrt{R}}+\alpha \sqrt{R})J_0(\mu R)+\mu\sqrt{R}J'_0(\mu R)\} + \\ + \ c_{q+l}\cdot\{(\frac{1}{2\sqrt{R}}+\alpha \sqrt{R})J_{-0}(\mu R)+\mu\sqrt{R}J'_{-0}(\mu R)\}=0, \ \textup{as} \ l=1,...,q_0. \\
c_l\cdot2^{\nu_l}\Gamma(1+\nu_l)\mu^{-\nu_l}\{(\frac{1}{2\sqrt{R}}+\alpha \sqrt{R})J_{\nu_l}(\mu R)+\mu\sqrt{R}J'_{\nu_l}(\mu R)\} + \\ + \ c_{q+l}\cdot2^{-\nu_l} \Gamma(1-\nu_l) \mu^{\nu_l}\{(\frac{1}{2\sqrt{R}}+\alpha \sqrt{R})J_{-\nu_l}(\mu R)+\mu\sqrt{R}J'_{-\nu_l}(\mu R)\}=0, \\ \ \textup{as} \ l=q_0+1,...,q,
\end{align*}
We can rewrite this system of equations in a compact form as follows 
$$
(J^+(\mu);J^-(\mu))\vec{\phi}=0,$$
where the matrices $J^{\pm}(\mu)$ are defined above.
\\[3mm] We obtain equations which have to be satisfied by the $\mu^2$-eigenvectors of the self-adjoint realization $\mathcal{L}$: 
\begin{align*}
\left(\begin{array}{cc}\mathcal{A} & \mathcal{B} \\ J^{+}(\mu) & J^{-}(\mu)\end{array}\right)\vec{\phi}=0.
\end{align*}
This equation has non-trivial solutions if and only if the determinant of the matrix in front of the vector is zero. Hence we finally arrive at the following implicit eigenvalue equation
\begin{align*}
F(\mu):=\textup{det}\left(\begin{array}{cc}\mathcal{A} & \mathcal{B} \\ J^{+}(\mu) & J^{-}(\mu)\end{array}\right)\stackrel{!}{=}0.
\end{align*}
\end{proof}
\begin{prop}\label{bahaviour-at-zero1} With $\nu_l=\sqrt{\lambda_l+1/4}$ and $\kappa = \frac{1}{2\sqrt{R}}+\alpha \sqrt{R}$
\begin{align*}F(0)=\textup{det}\left(\!\!\begin{array}{cc}\mathcal{A} & \mathcal{B} \\ \begin{array}{l}\kappa Id_{q_o} \qquad \qquad 0 \\ 0 \quad \textup{diag}(\kappa R^{\nu_l}+\nu_l R^{\nu_l-\frac{1}{2}})\end{array}& \begin{array}{l}(\kappa \log R + \frac{1}{\sqrt{R}})Id_{q_o} \qquad 0 \\ 0 \quad \textup{diag}(\kappa R^{-\nu_l}-\nu_l R^{-\nu_l-\frac{1}{2}})\end{array}\end{array}\!\!\right)\end{align*}
\end{prop}
\begin{proof}
The asymptotics \eqref{ass1}, \eqref{ass2} and \eqref{ass3}, where $O(x^2)$ is in fact power-series in $(x\mu)^2$ with no constant term, imply by straightforward computations:
\begin{align*}
\kappa J_0(\mu R)+\mu \sqrt{R} J'_0(\mu R)\rightarrow \kappa,  \ \textup{as} \ \mu \to 0, \\
\kappa J_{-0}(\mu R)+\mu \sqrt{R} J'_{-0}(\mu R) \rightarrow \kappa \cdot \log R + \frac{1}{\sqrt{R}},  \ \textup{as} \ \mu\to 0, \\
2^{\pm \nu_l}\Gamma(1\pm\nu_l)\mu^{\mp\nu_l}J_{\pm\nu_l}(\mu R)\to R^{\pm \nu_l}, \ \textup{as} \ \mu\to 0, \\
2^{\pm \nu_l}\Gamma(1\pm\nu_l)\mu^{\mp\nu_l}\mu\sqrt{R}J'_{\pm\nu_l}(\mu R) \to \pm \nu_l R^{\pm \nu_l-\frac{1}{2}}, \ \textup{as} \ \mu\to 0,
\end{align*}
where $l=q_0+1,..,q$. These relations prove the statement.
\end{proof} \ \\
\\[-7mm] The next proposition is similar to [KLP2, Proposition 4.3] and we use the notation therein.
\begin{prop}\label{asymptotics}
Let $\Upsilon \subset \C$ be a closed angle in the right half-plane. Then as $|x|\to \infty, x\in \Upsilon$ we can write  
\begin{align*}
F(ix)= \rho x^{|\nu|+\frac{q}{2}}e^{qxR}(2\pi)^{-\frac{q}{2}}(\widetilde{\gamma}-\log x)^{q_o}\times \\ p((\widetilde{\gamma} -\log   x)^{-1},x^{-1})\left(1+O\left(\frac{1}{x}\right)\right),
\end{align*}
where $\gamma$ is the Euler constant, $|\nu|=\nu_{q_o+1}+...+\nu_q$. Moreover, as $|x|\to \infty$ with $x\in \Upsilon$, $O(1/x)$ is a power-series in $x^{-1}$ with no constant term. Furthermore we have set:  
\begin{align*}\widetilde{\gamma}:=\log  2-\gamma, \ \rho:=\prod\limits_{l=q_o+1}^q2^{-\nu_l}\Gamma(1-\nu_l),& \\ 
p(x,y):=\textup{det}\left( \begin{array}{cc} \mathcal{A} & \mathcal{B} \\ \begin{array}{cc}x\cdot Id_{q_0} & 0 \\ 0 & \textup{diag} \ [\tau_ly^{2\nu_l}]\end{array}& Id_{q}\end{array}\right), \quad \textup{with}& \ \tau_l:=\frac{\Gamma(1+\nu_l)}{\Gamma(1-\nu_l)}2^{2\nu_l},
\end{align*} where the submatrix $\textup{diag} \ [\tau_ly^{2\nu_l}]$ has entries for $l=q_0+1,..,q$.
\end{prop}
\begin{proof}
We present $F(ix)$ in terms of modified Bessel functions of first and second kind. We use following well-known relations $$(iz)^{-\nu}J_{\nu}(iz)=z^{-\nu}I_{\nu}(z),\ J'_{\nu}(z)=J_{\nu-1}(z)-\frac{\nu}{z}J_{\nu}(z)$$ to analyze the building bricks of $F(ix)$ where we put with $l=q_0+1,..,q$
\begin{align*}
A_l^{\pm}:=2^{\pm \nu_l}\Gamma(1\pm\nu_l)(ix)^{\mp\nu_l}\left(\kappa J_{\pm\nu_l}(ix R) +ix\sqrt{R}J'_{\pm\nu_l}(ix R)\right)= 2^{\pm \nu_l}\Gamma(1\pm\nu_l) \\ \cdot \left((\kappa \mp \frac{\nu_l}{\sqrt{R}})x^{\mp \nu_l}I_{\pm \nu_l}(xR)+\sqrt{R}x^{\mp \nu_l+1}I_{\pm \nu_l-1}(xR)\right)
\end{align*}
\begin{align*}
B:=\kappa J_0(ixR)+\sqrt{R}ixJ'_0(ixR)=\kappa I_0(xR)+\sqrt{R}xI'_0(xR),
\end{align*}
and using the identity $J_{-0}(ixR)=-(\log x-\widetilde{\gamma})I_0(xR)-K_0(xR)$ from [KLP2, Section 4.3, p. 857] where $K_*$ denotes the modified Bessel function of second kind:
\begin{align*}
C:=\kappa J_{-0}(ixR)+\sqrt{R}ixJ'_{-0}(ixR)=\kappa J_{-0}(ixR)+\sqrt{R}\frac{d}{dR}J_{-0}(ixR)&= \\ =\kappa (-(\log  x-\widetilde{\gamma})I_0(xR)-K_0(xR))+\sqrt{R}\frac{d}{dR}(-(\log  x-\widetilde{\gamma})I_0(xR)&- \\ -K_0(xR)) = \kappa (-(\log  x-\widetilde{\gamma})I_0(xR)-K_0(xR))&+ \\+ \sqrt{R}(-(\log  x-\widetilde{\gamma})xI'_0(xR)-xK'_0(xR))&.
\end{align*}
Now in order to compute the asymptotics of $F(ix)$ we use following property of the Bessel functions: as $x\to \infty$ with $x \in \Upsilon$ we have by [AS, p. 377] $$I_{\nu}(x), I'_{\nu}(x)\sim \frac{e^x}{\sqrt{2\pi x}}\left(1+O(x^{-1})\right)$$ $$\Rightarrow \ \frac{I_{\nu}(xR)}{I_{-\nu}(xR)}\sim 1, \ \frac{I_{-\nu-1}(xR)}{I_{-\nu}(xR)}\sim 1, \ \frac{I_{\nu-1}(xR)}{I_{-\nu}(xR)}\sim 1,$$ where as $|x|\to \infty$ with $x\in \Upsilon$, $O(x^{-1})$ is a power-series in $x^{-1}$ with no constant term. From here we obtain the asymptotics of the terms $A_l^{\pm},B,C$ as $x\to \infty, x \in \Upsilon$ with the same meaning for $O(x^{-1})$:
\begin{align*}
A_l^+= 2^{-\nu_l}\Gamma(1-\nu_l) x^{\nu_l} \frac{e^{xR}}{\sqrt{2\pi xR}}\left[2^{2\nu_l}\frac{\Gamma (1+\nu_l)}{\Gamma (1-\nu_l)}x^{-2\nu_l}\right] \times \\ \left((\kappa - \frac{\nu_l}{\sqrt{R}})+  x\sqrt{R}\right)\cdot (1+O(x^{-1}))= \\
= 2^{-\nu_l}\Gamma(1-\nu_l) x^{\nu_l+1/2} \frac{e^{xR}}{\sqrt{2\pi}} \left[2^{2\nu_l}\frac{\Gamma (1+\nu_l)}{\Gamma (1-\nu_l)}x^{-2\nu_l}\right] \times \\  (1+O(x^{-1})).
\end{align*}
Similarly we compute
\begin{align*}
A_l^-= 2^{-\nu_l}\Gamma(1-\nu_l) x^{\nu_l+1/2} \frac{e^{xR}}{\sqrt{2\pi}}\cdot (1+O(x^{-1})), \\
B= \frac{e^{xR}}{\sqrt{2\pi xR}}(\kappa +x \sqrt{R})\cdot (1+O(x^{-1}))= \sqrt{x}\frac{e^{xR}}{\sqrt{2\pi}}\cdot (1+O(x^{-1})), \\ C= \frac{e^{xR}}{\sqrt{2\pi xR}}(\widetilde{\gamma}-\log  x)(\kappa +x \sqrt{R})\cdot (1+O(x^{-1}))=\\=\sqrt{x}\frac{e^{xR}}{\sqrt{2\pi}}(\widetilde{\gamma}-\log  x)\cdot (1+O(x^{-1})),
\end{align*}
where we have further used the fact that by [AS, p. 378] $K_0(xR)$ is exponentially decaying as $|x|\to \infty, x\in \Upsilon$. Now substitute these asymptotics into the definition of $F(ix)$ and obtain
\begin{align*}
F(ix) = \left[\prod\limits_{l=q_0+1}^q 2^{-\nu_l}\Gamma(1-\nu_l) x^{\nu_l+\frac{1}{2}}\right] \left(\frac{e^{xR}}{\sqrt{2\pi}}\right)^q x^{q_0/2}(\widetilde{\gamma}-\log  x)^{q_0}\times \\ \textup{det}\left( \begin{array}{cc} \mathcal{A} & \mathcal{B} \\ \begin{array}{cc}(\widetilde{\gamma}-\log  x)^{-1}Id_{q_0} & 0 \\ 0 & 
* \end{array}& Id_{q}\end{array} \right) (1+O(x^{-1})), \\ \textup{where} \ *=\textup{diag}\left(2^{2\nu_l}\frac{\Gamma (1+\nu_l)}{\Gamma (1-\nu_l)}x^{-2\nu_l}\right).
\end{align*}
\end{proof} \ \\
Using the expansion in [KLP1, (4.9)] we evaluate the asymptotics of $p((\widetilde{\gamma}-\log  x)^{-1},x^{-1})$ and obtain in the notation introduced in the statement of Proposition \ref{asymptotics} 
\begin{align}
F(ix)= a_{j_o\alpha_o}\rho x^{|\nu|+\frac{q}{2}-2\alpha_o}\left(\frac{e^{xR}}{\sqrt{2\pi}}\right)^{q} (\widetilde{\gamma}-\log  x)^{q_o-j_o}\left(1+G(x)\right), \label{G-function}
\end{align}
where $G(x)=O(\frac{1}{\log(x)})$ and $G'(x)=O(\frac{1}{x\log^2(x)})$ as $|x| \to \infty$ with $x$ inside any fixed closed angle of the right half plane of $\C$. The coefficients $\alpha_0, j_0,a_{j_o\alpha_o}$ are defined in [KLP2, (2.2)] and are characteristic values of the boundary conditions $(\mathcal{A},\mathcal{B})$ at the cone singularity. We recall their definition here for convenience. 
\begin{defn}\label{a-j-a}
The expression $p(x,y)$ defined in Proposition \ref{asymptotics} can be written as a finite sum $$p(x,y)=\sum a_{j\A}x^jy^{\A}.$$ The characteristic values $\A_0,j_0,a_{\A_oj_o}$ are defined as follows:
\begin{enumerate}
\item The coefficient $\A_0$ is the smallest of all exponents $\A$ with $a_{j\A}\neq 0$.
\item The coefficient $j_0$ is the smallest of all exponents $j$ with $a_{j\A_o}\neq 0$.
\item The coefficient $a_{j_o\A_o}$ is the coefficient in the polynomial $p(x,y)$ of the summand $x^{j_o}y^{\A_o}$. 
\end{enumerate}
\end{defn} \ \\
\\[-7mm] Unfortunately the asymptotic expansion, obtained in Proposition \ref{asymptotics}, does not hold uniformly for arguments $z$ of $F(z)$ in a fixed closed angle of the positive real axis. This gap is closed by the following proposition.
\begin{prop}\label{F-real-asymptotics}
Fix any $\theta \in [0,\pi)$ and put $\Omega:=\{z \in \C | |\textup{arg}(z)|\leq \theta\}$. Then for $|z|\to \infty, z\in \Omega$ we have the following uniform expansion:
\begin{align*}
F(z)=\prod\limits_{l=q_0+1}^q\left\{ 2^{-\nu_l}\Gamma(1-\nu_l)z^{\nu_l+1/2}\sqrt{\frac{2}{\pi}}\cos (zR+\frac{\nu_l\pi}{2}+\frac{\pi}{4})\right\} \times \\
\times \left\{\sqrt{\frac{2z}{\pi}}(\log z -\widetilde{\gamma})\cos (zR-\frac{3}{4}\pi)\right\}^{q_0}\cdot \det M(z).
\end{align*}
Here the matrix $M(z)$ is given as follows:
\begin{align*}
M(z)=\left(\begin{array}{cc} \mathcal{A} & \mathcal{B} \\ \begin{array}{cc}b(z)\cdot Id_{q_0} & 0 \\ 0 & \textup{diag} \ [a^+_l(z)]\end{array} & \begin{array}{cc}c(z)\cdot Id_{q_0} & 0 \\ 0 & \textup{diag} \ [a^-_l(z)]\end{array} \end{array} \right),                                              
\end{align*}
where for $l=q_0+1,..,q$ we have
\begin{align*}
&a^+_l(z)=2^{2\nu_l}\frac{\Gamma(1+\nu_l)}{\Gamma(1-\nu_l)}z^{-2\nu_l}\frac{\cos(zR-\frac{\nu_l\pi}{2}+\frac{\pi}{4})}{\cos(zR+\frac{\nu_l\pi}{2}+\frac{\pi}{4})}
\cdot \left(1+\frac{f^+_l(z)}{\cos(zR-\frac{\nu_l\pi}{2}+\frac{\pi}{4})}\right), \\
&a^-_l(z)=1+\frac{f^-_l(z)}{\cos(zR+\frac{\nu_l\pi}{2}+\frac{\pi}{4})}, \ 
b(z)=\frac{1}{\widetilde{\gamma}-\log z} \cdot \left(1+\frac{f_b(z)}{\cos(zR-\frac{3}{4}\pi)}\right), \\
&c(z)=1+\frac{f_c(z)}{\cos(zR-\frac{3}{4}\pi)},
\end{align*}
and the functions $f^{\pm}_l(z), f_b(z),f_c(z)$ have the following asymptotic behaviour as $|z|\to \infty, z\in \Omega$
\begin{align*}
f^{\pm}_l(z)=e^{|\textup{Im}(zR)|}O\left(\frac{1}{|z|}\right), &\ \frac{d}{dz}f^{\pm}_l(z)=e^{|\textup{Im}(zR)|}O\left(\frac{1}{|z|}\right), \\
f_b(z)=e^{|\textup{Im}(zR)|}O\left(\frac{1}{|z|}\right), &\ \frac{d}{dz}f_b(z)=e^{|\textup{Im}(zR)|}O\left(\frac{1}{|z|}\right), \\
f_c(z)=e^{|\textup{Im}(zR)|}O\left(\frac{1}{|\log z|}\right), &\ \frac{d}{dz}f_c(z)=e^{|\textup{Im}(zR)|}O\left(\frac{1}{|\log z|}\right). \\
\end{align*}
\end{prop}
\begin{proof}
The formulas [AS, 9.2.1, 9.2.2] provide the standard asymptotic behaviour of Bessel functions as $|z|\to \infty, z\in \Omega$
\begin{align*}
J_{\nu}(z)=\sqrt{\frac{2}{\pi z}}\left(\cos (z-\frac{\nu \pi}{2}-\frac{\pi}{4})+f(z)\right), \ f(z)=e^{|\textup{Im}(z)|}O\left(\frac{1}{|z|}\right), \\
Y_{\nu}(z)=\sqrt{\frac{2}{\pi z}}\left(\sin (z-\frac{\nu \pi}{2}-\frac{\pi}{4})+g(z)\right), \ g(z)=e^{|\textup{Im}(z)|}O\left(\frac{1}{|z|}\right).
\end{align*}
Here $\nu \in \R$ and the expansions are uniform in the closed angle $\Omega$. Moreover we infer from the more explicit form of asymptotics in [GRA, 8.451]:
$$\frac{d}{dz}f(z)=e^{|\textup{Im}(z)|}O\left(\frac{1}{|z|}\right), \ \frac{d}{dz}g(z)=e^{|\textup{Im}(z)|}O\left(\frac{1}{|z|}\right).$$
We apply these asymptotics in order to analyze the asymptotic behaviour as $|z|\to \infty, z\in \Omega$ of the following building bricks of $F(z)$:
\begin{align*}
A^{\pm}_l &:=2^{\pm\nu_l}\Gamma (1\pm \nu_l)z^{\mp \nu_l}\left(\kappa J_{\pm \nu_l}(zR)+z\sqrt{R}J'_{\pm \nu_l}(zR)\right), \ l=q_0+1,..,q, \\
B &:=\kappa J_0(zR)+z\sqrt{R}J'_0(zR), \\
C &:=\kappa J_{-0}(zR)+z\sqrt{R}J'_{-0}(zR).
\end{align*}
Straightforward application of the asymptotics for $J_{\nu}(z)$ and $Y_{\nu}(z)$ as $|z|\to \infty, z\in \Omega$ and furthermore the use of the well-known formulas 
\begin{align*}
J'_{\nu}(z)=J_{\nu-1}(z)-\frac{\nu}{z}J_{\nu}(z),\\
J'_0(z)=-J_1(z), \ Y'_0(z)=-Y_1(z),
\end{align*}
lead to the following intermediate results:
\begin{align*}
&A^+_l=2^{\nu_l}\Gamma(1+\nu_l)z^{-\nu_l+1/2}\sqrt{\frac{2}{\pi}}\cos(zR-\frac{\nu_l\pi}{2}+\frac{\pi}{4})
\cdot \left(1+\frac{f^+_l(z)}{\cos(zR-\frac{\nu_l\pi}{2}+\frac{\pi}{4})}\right), \\
&A^-_l=2^{-\nu_l}\Gamma(1-\nu_l)z^{\nu_l+1/2}\sqrt{\frac{2}{\pi}}\cos (zR+\frac{\nu_l\pi}{2}+\frac{\pi}{4})\cdot \left(1+\frac{f^-_l(z)}{\cos(zR+\frac{\nu_l\pi}{2}+\frac{\pi}{4})}\right), \\
&B=-\sqrt{\frac{2z}{\pi}}\cos (zR-\frac{3}{4}\pi)\cdot\left(1+\frac{f_b(z)}{\cos(zR-\frac{3}{4}\pi)}\right), \\
&C= \sqrt{\frac{2z}{\pi}}(\log z -\widetilde{\gamma})\cos (zR-\frac{3}{4}\pi)\cdot \left(1+\frac{f_c(z)}{\cos(zR-\frac{3}{4}\pi)}\right),
\end{align*}
where the functions $f^{\pm}_l(z), f_b(z)$ and their derivatives are of the asymptotics $e^{|\textup{Im}(zR)|}O\left(1/|z|\right)$ as $|z|\to \infty, z\in \Omega$. The function $f_c(z)$ and its derivative are of the asymptotics $e^{|\textup{Im}(zR)|}O\left(1/|\log z|\right)$, as $|z|\to \infty, z\in \Omega$. Recall finally the definition of $F(z)$:
\begin{align*}
F(z)=\det\left(\begin{array}{cc} \mathcal{A} & \mathcal{B} \\ \begin{array}{cc}B\cdot Id_{q_0} & 0 \\ 0 & \textup{diag} \ [A^+_l]\end{array} & \begin{array}{cc}C\cdot Id_{q_0} & 0 \\ 0 & \textup{diag} \ [A^-_l]\end{array} \end{array} \right).                                              
\end{align*}
Inserting the asymptotics for $A^{\pm}_l, B$ and $C$ into the definition of $F(z)$ we obtain the statement of the proposition.
\end{proof} \ \\
\\[-7mm] The following result on the spectrum of $\mathcal{L}$ is a corollary of Proposition \ref{asymptotics} and is necessary for the definition and discussion of certain contour integrals below.
\begin{cor}\label{finite-imaginary-zeros}
The self-adjoint operator $\mathcal{L}$ is bounded from below. The zeros of its implicit eigenvalue function $F(\mu)$ are either real or purely imaginary, where the number of the purely imaginary zeros is finite.
\\[3mm] The positive eigenvalues of $\mathcal{L}$ are given by squares of the positive zeros of $F(\mu)$. The negative eigenvalues of $\mathcal{L}$ are given by squares of the purely imaginary zeros of $F(\mu)$ with positive imaginary part, i.e. counting the eigenvalues of $\mathcal{L}$ and zeros of $F(\mu)$ with their multiplicities we have
\begin{align}\label{zeros-F}
\textup{Spec}\mathcal{L}\backslash \{0\}=\{\mu^2\in \R | F(\mu)=0, \mu>0 \wedge \mu=ix,x>0\}
\end{align}
\end{cor}
\begin{proof}
The relation between zeros of $F(\mu)$ and eigenvalues of $\mathcal{L}$ is established in Proposition \ref{implicit-eigenvalue-equation}. The self-adjoint operator $\mathcal{L}$ has real spectrum, hence the zeros of $F(\mu)$ are either real or purely imaginary, representing positive or negative eigenvalues of $\mathcal{L}$, respectively.
\\[3mm] The standard infinite series representation of Bessel functions (see [AS, p.360]) implies that zeros of $F(\mu)$ are symmetric about the origin and any two symmetric zeros do not correspond to two linearly independent eigenfunctions of $\mathcal{L}$. Hence the non-zero eigenvalues of $\mathcal{L}$ are in one-to-one correspondence with zeros of $F(\mu)$ at the positive real and the positive imaginary axis.
\\[3mm] The asymptotics \eqref{G-function} implies in particular that depending on the characteristic values $j_0,q_0,a_{j_o\A_o}$ of the boundary conditions $(\mathcal{A},\mathcal{B})$, the implicit eigenvalue function $F(ix)$ goes either to plus or minus infinity as $x\in \R, x\to \infty$ and cannot become zero for $|x|$ sufficiently large. Since the zeros of the meromorphic function $F(\mu)$ are discrete, we deduce that $F(\mu)$ has only finitely many purely imaginary eigenvalues. Thus in turn, $\mathcal{L}$ has only finitely many negative eigenvalues, i.e. is bounded from below.
\end{proof}\ \\
\\[-8mm] Next we fix an angle $\theta \in (0,\pi /2)$ and put for any $a\in \R^+$:\begin{align*}
&\delta(a):=\{z \in \C | \textup{Re}(z)=a, |\arg(z)|\leq \theta\}, \\
&\rho(a):=\{z \in \C | |z|=a/\cos (\theta), |\arg(z)|\in [\theta, \pi /2 ]\}, \\
&\gamma(a):=\delta(a)\cup \rho(a),
\end{align*}
where the contour $\gamma(a)$ is oriented counter-clockwise, as in the Figure \ref{contour-vertman} below:
\begin{figure}[h]
\begin{center}
\includegraphics[width=0.4\textwidth]{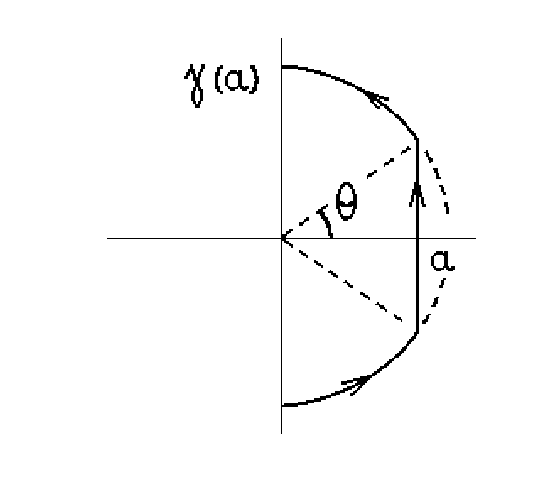}
\vspace{-3mm} \caption{The contour $\gamma(a)$ for the fixed $\theta \in (0,\pi /2)$ and $a\in \R^+$.}
\label{contour-vertman}
\end{center}
\end{figure}
\\[-2mm] Furthermore we fix the branch of logarithm in $\C\backslash \R^+$ with $0\leq Im \log (z) < 2\pi$. In this setup, the following result is a central application of the asymptotic expansions in Proposition \ref{asymptotics} and Proposition \ref{F-real-asymptotics}.
\begin{prop}\label{contour-integral-method}
There exists a sequence $(a_n)_{n\in \N}$ of positive real numbers with $a_n\to \infty$ as $n\to \infty$, such that $F(a_n)\neq 0$ for all $n \in \N$ and for $Re(s)>1/2$ the following integrals $$\int_{\gamma(a_n)}z^{-2s}\frac{d}{dz}\log F(z) dz, \ n \in \N$$ are well-defined and the sequence of integrals converges to zero as $n\to \infty$. 
\end{prop}
\begin{proof}
Consider first the logarithmic form of the asymptotics \eqref{G-function} 
\begin{align*}
\log F(ix)=\log \left(a_{j_o\A_o}\cdot \rho \cdot (2\pi)^{-q/2}\right) + (|\nu|+\frac{q}{2}-2\A_0)\log x + q x R + \\
+ (q_0-j_0)\log (\widetilde{\gamma}-\log x) +\log (1+G(x)),
\end{align*}
where $G(x)=O(\frac{1}{\log(x)})$ and $G'(x)=O(\frac{1}{x\log^2(x)})$ as $|x| \to \infty$ with $ix\in \{z \in \C| |\arg(z)|\in [\theta, \pi /2], \ \textup{Im}(z)>0\}$. Same asymptotics holds for $ix\in \{z \in \C| |\arg(z)|\in [\theta, \pi /2], \ \textup{Im}(z)<0\}$, since $F(ix)=F(-ix)$ by the standard infinite series representation of Bessel functions [AS, p.360]. By straightforward calculations we see for Re$(s)>1/2$:
\begin{align}\label{P-Loya}
\int_{\rho(a_n)}z^{-2s}\frac{d}{dz}\log F(z) dz \xrightarrow{n\to \infty} 0,
\end{align}
for any sequence $(a_n)_{n\in \N}$ of positive real numbers with $a_n\to \infty$ as $n\to \infty$. Thus it remains to find a sequence $(a_n)_{n\in \N}\subset \R^+$ which goes to infinity and further ensures that 
\begin{align}
\int_{\delta(a_n)}z^{-2s}\frac{d}{dz}\log F(z) dz \xrightarrow{n\to \infty} 0,
\end{align}
where for each $n\in \N$ the integral is well-defined. In order to construct such a sequence, fix $a>0$ subject to the following conditions
\begin{align}
\label{condition1} \cos (aR\pm \frac{\nu_l \pi}{2}+\frac{\pi}{4})\neq 0, \ l=q_0+1,..,q;\\
\label{condition2} \cos (aR-\frac{3}{4}\pi)\neq 0.
\end{align}
Such a choice is always possible, due to discreteness of zeros of the holomorphic function $\cos (z)$. Given such an $a>0$, we define $$\Delta(a):=\bigcup_{k\in \N}\delta(a+\frac{2\pi}{R}k).$$ Using $\cos(z)=(e^{iz}+e^{-iz})/2$ we find for any $\xi \in \R$ with $\cos(aR+\xi)\neq 0$ as $|z|\to \infty, z\in \Delta(a)$
\begin{align}\label{cos-k1}
\cos (zR+\xi)=e^{|\textup{Im}(zR)|}O(1),
\end{align}
where $|O(1)|$ is bounded away from zero with the bounds depending only on the sign of Im$(zR)$, $a>0$ and $\xi \in \R$. Putting $\vec \A=(\A_{q_0+1},..,\A_q)\in \{0,1\}^{q_1}, q_1=q-q_0$, we obtain for the asymptotic behaviour of $\det M(z)$, introduced in Proposition \ref{F-real-asymptotics}, as $|z|\to \infty, z\in \Delta(a)$:
\begin{align*}
\det M(z)=\sum_{j=0}^{q_0}\sum_{\vec \A\in \{0,1\}^{q_1}}\sum_{\beta=0}^{q}\textup{const}(j,\vec \A,\beta)\left[\frac{1}{\widetilde{\gamma}-\log z}\left(1+O\left(\frac{1}{|z|}\right)\right)\right]^j\times \\ \times
\prod_{l=q_0+1}^{q}\left[z^{-2\nu_l}\frac{\cos(zR-\frac{\nu_l\pi}{2}+\frac{\pi}{4})}{\cos(zR+\frac{\nu_l\pi}{2}+\frac{\pi}{4})}
\cdot \left(1+O\left(\frac{1}{|z|}\right)\right)\right]^{\A_l} \cdot \left[1+O\left(\frac{1}{|\log z|}\right)\right]^{\beta},
\end{align*}
where $\textup{const}(j,\vec \A,\beta)$ depends moreover on $\mathcal{A}$ and $\mathcal{B}$. In fact one has by construction
\begin{align*}
\sum_{\vec \A \in I_{\A}} \sum_{\beta=0}^q \textup{const}(j,\vec \A, \beta)=a_{j\A},
\end{align*}
where $I_{\A}=\{\vec \A \in \{0,1\}^{q_1}| \sum_{l=q_o+1}^q\nu_l\A_l=\A\}$ and $a_{j\A}$ are the coefficients in the Definition \ref{a-j-a}. Multiplying out the expression for $\det M(z)$ we compute:
\begin{align*}
&\det M(z)=\sum_{j=0}^{q_0}\sum_{\vec \A\in \{0,1\}^{q_1}}\sum_{\beta=0}^{q}\textup{const}(j,\vec \A,\beta)\left[\frac{1}{\widetilde{\gamma}-\log z}\right]^j\times \\ &\times
\prod_{l=q_0+1}^{q}\left[z^{-2\nu_l}\frac{\cos(zR-\frac{\nu_l\pi}{2}+\frac{\pi}{4})}{\cos(zR+\frac{\nu_l\pi}{2}+\frac{\pi}{4})}\right]^{\A_l} \cdot \left[1+f_{j,\vec \A,\beta}(z)\right], \ f_{j,\vec \A,\beta}(z)=O\left(\frac{1}{|\log z|}\right).
\end{align*}
The asymptotic behaviour of $f_{j,\vec \A,\beta}(z)$ under differentiation follows from Proposition \ref{F-real-asymptotics}
\begin{align}\label{cos-k3}
\frac{d}{dz}f_{j,\vec \A,\beta}(z)=O\left(\frac{1}{|\log z|}\right).
\end{align}
Before we continue let us make an auxiliary observation, in the spirit of \eqref{cos-k1}. Under the condition \eqref{condition1} on the choice of $a>0$, we have for $z\in \Delta(a)$ and $l=q_0+1,..,q$:
\begin{align}\label{cos-k}
&\frac{\cos(zR - \frac{\nu_l \pi}{2}+\frac{\pi}{4})}{\cos(zR + \frac{\nu_l \pi}{2}+\frac{\pi}{4})}=C\cdot \left(1+\frac{e^{-2|\textup{Im}(zR)|}C'}{1+e^{-2|\textup{Im}(zR)|}C''}\right), 
\end{align}
where the constants $C,C',C''$ are given explicitly as follows:
\begin{align*}
&C=\exp \left(i \, \textup{sign}[\textup{Im}(z)](\nu_l \pi)\right), \\
&C'=\exp \left(i \, \textup{sign}[\textup{Im}(z)] (2aR-\nu_l\pi+\frac{\pi}{2})\right)-\exp \left(i \, \textup{sign}[\textup{Im}(z)] (2aR+\nu_l\pi+\frac{\pi}{2})\right), \\
&C''=\exp \left(i \, \textup{sign}[\textup{Im}(z)] (2aR+\nu_l\pi+\frac{\pi}{2})\right).
\end{align*}
Note that the constants are non-zero, depend only on $\textup{sign}[\textup{Im}(z)]$, the choice of $a$ and $\nu_l$. Hence for $|\textup{Im}(zR)|\to \infty$ the quotient \eqref{cos-k} tends to $C\neq 0$. Therefore, due to conditions \eqref{condition1} and \eqref{condition2}, there exist constants $\mathfrak{C}_1>0$ and $\mathfrak{C}_2>0$, depending only on $a$, such that for $z\in \Delta(a)$ and for all $l=q_0+1,..,q$ we have:
\begin{align}\label{cos-k4}
\mathfrak{C}_1 \leq \left|\frac{\cos(zR- \frac{\nu_l \pi}{2}+\frac{\pi}{4})}{\cos(zR+ \frac{\nu_l \pi}{2}+\frac{\pi}{4})}\right| \leq \mathfrak{C}_2.
\end{align}
In particular the cosinus terms in $\det M(z)$ are not relevant for its asymptotic behaviour as $|z|\to \infty, z\in \Delta(a)$. Now let us consider the summands in $\det M(z)$ of slowest decrease as $|z|\to \infty, z \in \Delta(a)$:
\begin{align*}
\left[\frac{1}{\widetilde{\gamma}-\log z}\right]^{j_0} z^{-2\A_0} \cdot & \left\{\sum_{\beta=0}^{q}\sum_{\vec \A\in I_{\A_o}}\textup{const}(j_0,\vec\A,\beta) \prod_{l=q_0+1}^{q}\left[\frac{\cos(zR-\frac{\nu_l\pi}{2}+\frac{\pi}{4})}{\cos(zR+\frac{\nu_l\pi}{2}+\frac{\pi}{4})}\right]^{\A_l}\right\}\\
=:&\left[\frac{1}{\widetilde{\gamma}-\log z}\right]^{j_0} z^{-2\A_0}g(z),
\end{align*}
where the coefficients $j_0,\A_0$ correspond to those in Definition \ref{a-j-a}. By similar calculus as behind \eqref{cos-k} we can write 
\begin{align*}
g(z)=\widetilde{C}\left(1+\frac{e^{-2|\textup{Im}(zR)|}C'(a,z)}{1+e^{-2|\textup{Im}(zR)|}C''(a,z)}\right),
\end{align*}
where $C'(a,z), C''(a,z)$ further depend on $\mathcal{A}, \mathcal{B}$ and $\nu_l, l=q_0+1,..,q$. Moreover they are bounded from above independently of $a>0$ and $z\in \Delta(a)$. The factor $\widetilde{C}$ is given explicitly as follows: 
\begin{align*}
\widetilde{C}=\sum_{\vec \A \in I_{\A_o}} \sum_{\beta=0}^q \textup{const}(j_0,\vec \A, \beta)\cdot \exp (i \, \textup{sign}[\textup{Im}(z)]\pi \A_0) \\
= a_{j_o \A_o}\cdot \exp (i \, \textup{sign}[\textup{Im}(z)]\pi \A_0)\neq 0,
\end{align*}
since $a_{j_o\A_o}\neq 0$ by the definition of characteristic values in Definition \ref{a-j-a}. Since $g(z)$ is a meromorphic function with discrete zeros and poles, we can choose $a>0$ sufficiently large, still subject to conditions \eqref{condition1} and \eqref{condition2}, such that $g(z)$ has no zeros and poles on $\delta(a)$ and 
$$\left| \frac{e^{-2|\textup{Im}(zR)|}C'(a,z)}{1+e^{-2|\textup{Im}(zR)|}C''(a,z)} \right| << 1,$$ for $z \in \delta(a)$ with the highest possible absolute value of its imaginary part, i.e. with $|\textup{Im}(z)|=a\cdot \tan \theta$. This guarantees that there exist constants $\mathfrak{C}'_1>0$ and $\mathfrak{C}'_2>0$, depending only on $a>0$, such that for $z\in \Delta(a)$
\begin{align}\label{cos-k5}
\mathfrak{C}'_1 \leq |g(z)| \leq \mathfrak{C}'_2.
\end{align}
By similar arguments we find that $|\frac{d}{dz}g(z)|$ is bounded from above for $z \in \Delta(a)$. Using \eqref{cos-k5} we finally obtain for $\det M(z)$ as $|z|\to \infty, z\in \Delta(a)$
\begin{align*}
\det M(z)=\left[\frac{1}{\widetilde{\gamma}-\log z}\right]^{j_0} z^{-\A_0}g(z)(1+f(z)), f(z)=O\left(\frac{1}{|\log (z)|}\right),
\end{align*}
as $|z|\to \infty, z\in \Delta(a)$. Using \eqref{cos-k3}, \eqref{cos-k4} and boundedness of $g(z),g'(z)$ we obtain 
\begin{align*}
\frac{d}{dz}f(z)=O\left(\frac{1}{|\log z|}\right).
\end{align*}
In total we have derived the following asymptotic behaviour of $F(z)$ as $|z|\to \infty, z\in \Delta(a)$:
\begin{align*}
&F(z)=\prod\limits_{l=q_0+1}^q\left\{ 2^{-\nu_l}\Gamma(1-\nu_l)z^{\nu_l+1/2}\sqrt{\frac{2}{\pi}}\cos (zR+\frac{\nu_l\pi}{2}+\frac{\pi}{4})\right\} \times \\
&\times \left\{\sqrt{\frac{2z}{\pi}}(\log z -\widetilde{\gamma})\cos (zR-\frac{3}{4}\pi)\right\}^{q_0}\left[\frac{1}{\widetilde{\gamma}-\log z}\right]^{j_0} z^{-\A_0}g(z)(1+f(z)),
\end{align*}
where there exist positive constants $\mathfrak{C}'_1, \mathfrak{C}'_2, \mathfrak{C}''$, depending only on $a>0$, such that
\begin{align*}
\mathfrak{C}'_1 \leq |g(z)| \leq \mathfrak{C}'_2, \ |g'(z)| \leq \mathfrak{C}'', \\
f(z)=O\left(\frac{1}{|\log z|}\right), \ \frac{d}{dz}f(z)=O\left(\frac{1}{|\log z|}\right).
\end{align*}
Note that for $N\in \N$ sufficiently large, the asymptotics above, together with the conditions \eqref{condition1}, \eqref{condition2} and \eqref{cos-k5}, imply that $F(a+2\pi k /R)\neq 0$ for all $k\in \N, k\geq N$ (note also that by construction $\sum_{l=q_o+1}^q\nu_l+q_1/2-\A_0>0$). Putting $a_n:=a+2\pi (N+n)/R, n \in \N$ we obtain a sequence $(a_n)_{n\in \N}$ of positive numbers, going to infinity as $n\to \infty$ and we infer from the asymptotics of $F(z)$ above, that for Re$(s)>1/2$ 
\begin{align}
\int_{\delta(a_n)}z^{-2s}\frac{d}{dz}\log F(z) dz \xrightarrow{n\to \infty} 0,
\end{align}
where by construction for each $n\in \N$ we have $F(a_n)\neq 0$, and hence the integrals are well-defined. Together with $\eqref{P-Loya}$ this finally proves the statement of the proposition.
\end{proof} \ \\
\\[-7mm] The asymptotics obtained in Proposition \ref{asymptotics} implies that the contour integral $$\frac{1}{2\pi i}\int_{\gamma}\mu^{-2s}\frac{d}{d \mu}\log F(\mu)d\mu $$ with the fixed branch of logarithm in $\C \backslash \R^+$ such that $0 \leq Im \log z < 2\pi$ and the contour $\gamma$ defined in Figure \ref{contour-loya} below, converges for $Re(s)>1/2$.
\pagebreak

\begin{figure}[h]
\begin{center}
\includegraphics[width=0.5\textwidth]{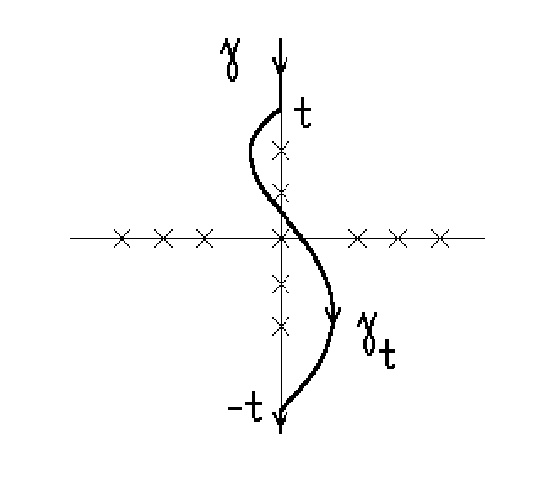}
\caption{The contour $\gamma$. The $\times$'s represent the zeros of $F(\mu)$. The number of purely imaginary zeros is finite by Corollary \ref{finite-imaginary-zeros}. The $t\in i\R$ is chosen such that $|t|^2$ is larger than the largest absolute value of negative eigenvalues of $\mathcal{L}$ (if present). The contour $\gamma_t\subset \gamma$ goes from $t$ to $-t$.}
\label{contour-loya}
\end{center}
\end{figure} 
The definition of $\gamma$ corresponds to [KLP1, Figure 1]. We can view the contour $\gamma$ to be closed up at infinity on the right hand side of $\C$. 
Then by construction $\gamma$ encircles the relevant zeros of $F(\mu)$ in \eqref{zeros-F}. As a consequence of Proposition \ref{contour-integral-method} we can apply the Argument Principle and finally arrive at
\begin{align*}
\zeta_{\mathcal{L}}(s)=\frac{1}{2\pi i}\int_{\gamma}\mu^{-2s}\frac{d}{d \mu}\log F(\mu)d\mu, \quad Re(s) > 1/2.
\end{align*}
This integral representation of the zeta-function is referred by K. Kirsten, P. Loya and J. Park as the "Contour integral method". Thus, on the basis of Proposition \ref{contour-integral-method}, we have verified applicability of the Contour integral method in the regular-singular setup, which is the basis for further arguments in [KLP1] and [KLP2].
\\[3mm] Breaking the integral into three parts $\gamma = \{i x | x\geq t\} \cup \gamma_t \cup \{i x | x \leq -t\}$ we obtain as in [KLP1, (4.10)]
\begin{align}
\zeta_{\mathcal{L}}(s)=\frac{\sin (\pi s)}{\pi }\int_{|t|}^{\infty}x^{-2s}\frac{d}{d x}\log F(ix)dx + \frac{1}{2\pi i}\int_{\gamma_t}\mu ^{-2s}\frac{F'(\mu)}{F(\mu)}d\mu. 
\end{align}
Analytic continuation of the first integral to $s=0$, see [KLP1, (4.12)] allows computation of the functional determinant of $\mathcal{L}$ after subtracting possible logarithmic singularities. We have the following result.

\begin{prop}\label{general-result} Under the assumption that $\ker \mathcal{L}=\{0\}$ we obtain in the notation of Propositions \ref{bahaviour-at-zero1} and \ref{asymptotics}
\begin{align*}
\exp \left[ - \lim_{s\to 0+}\frac{d}{ds}\left\{\frac{1}{2\pi i}\int_{\gamma}\mu^{-2s}\frac{d}{d\mu}\log   F(\mu)d\mu-(j_0-q_0)s\cdot \log  (s)\right\}\right]=\\ =\frac{(2\pi)^{q/2}}{a_{j_o\A_o}}(-2e^{\gamma})^{q_0-j_0}\prod_{l=q_0+1}^q\frac{2^{\nu_l}}{\Gamma(1-\nu_l)} \times \\ \textup{det}\left(\!\!\begin{array}{cc}\mathcal{A} & \mathcal{B} \\ \begin{array}{l}\kappa Id_{q_o} \qquad \qquad 0 \\ 0 \quad \textup{diag}(\kappa R^{\nu_l}+\nu_l R^{\nu_l-\frac{1}{2}})\end{array}& \begin{array}{l}(\kappa \log R + \frac{1}{\sqrt{R}})Id_{q_o} \qquad 0 \\ 0 \quad \textup{diag}(\kappa R^{-\nu_l}-\nu_l R^{-\nu_l-\frac{1}{2}})\end{array}\end{array}\!\!\right).
\end{align*}
\end{prop}
\begin{proof} Put $C:=a_{j_o\A_o}\cdot \rho \cdot (2\pi)^{-q/2}$ and rewrite the asymptotic expansion \eqref{G-function} for $|x|\to \infty$ with $x$ inside any fixed closed angle of the right half plane of $\C$ in its logarithmic form:
\begin{align}
\log F(ix)=\log C + (|\nu|+\frac{q}{2}-2\A_0)\log x + q x R + \nonumber \\
+ (q_0-j_0)\log (\widetilde{\gamma}-\log x) +\log (1+G(x)). \label{G-F}
\end{align}
In fact the asymptotics differs from the result in [KLP1, Proposition 4.3] only by a presence of an additional summand: $$\log (x\sqrt{R})^q.$$
Hence the same computations as those in [KLP1] give:
\begin{align*}
&\zeta (s, \mathcal{L}) - (j_0-q_0) s \log s = \\ = &\frac{\sin \pi s }{\pi }\left(|\nu|+\frac{q}{2}-2\alpha_0\right)\frac{|t|^{-2s}}{2s}+ \frac{\sin \pi s }{\pi }qR\frac{|t|^{-2s+1}}{2s-1}+ \\ + &\frac{\sin \pi s }{\pi }(j_0-q_0)g(s)+  \frac{\sin \pi s }{\pi }\int_{|t|}^{\infty}x^{-2s}\frac{d}{dx}\log  (1+G(x))dx+ \\ + &\frac{1}{2\pi i}\int_{\gamma_t}\mu^{-2s}\frac{F'(\mu)}{F(\mu)}d\mu,
\end{align*} 
where with [KLP1, (4.11)] the function $g(s)$ is entire and $g(0)=\gamma+\log  (2(\log  |t|-\widetilde{\gamma}))$. Explicit differentiation at $s\to 0+$ 
 leads to the following result (compare [KLP1, p.113]):
\begin{align*}
\lim_{s\to 0+}\frac{d}{ds}&\left\{\frac{1}{2\pi i}\int_{\gamma}\mu^{-2s}\frac{d}{d\mu}\log   F(\mu)d\mu-(j_0-q_0)s\cdot \log  (s)\right\}\\ =& -\left(|\nu|+\frac{q}{2}-2\alpha_0\right)\log  |t|-qR|t|+(j_0-q_0)\left(\gamma+\log  (2(\log  |t|-\widetilde{\gamma}))\right)\\-& \ \log  (1+G(|t|))-\frac{1}{\pi i}\int_{\gamma_t}\log  \mu\frac{F'(\mu)}{F(\mu)}d\mu =:Q.
\end{align*}
Using \eqref{G-F} we can evaluate $\log   (1+G(|t|))$ and by inserting it into the expression above we obtain 
\begin{align}\nonumber 
Q=- \!\log  \left(\frac{F(i|t|)}{C (-1)^{q_o-j_o}}\right)\!+\!(j_0-q_0)(\gamma +\log   2) - \\ - \frac{1}{\pi i}\int\limits_{\gamma_t}\!\log  \mu\frac{F'(\mu)}{F(\mu)}. \label{result2}
\end{align}
The formula above is a priori derived for $t=i|t|$ being on the upper-half of the imaginary axis. At this point we continue with the trick of [KLP1, Figure 2] to take $|t|\to 0$, which works well under the assumption $\ker \mathcal{L}=\{0\}$. 
\\[3mm] The integral over the finite contour $\gamma_t$ in \eqref{result2} vanishes as $t\to 0$. By triviality of $\ker \mathcal{L}$ we have $F(0)\neq0$ and obtain 
\begin{align}
Q=- \!\log  \left(\frac{F(0)}{C (-1)^{q_o-j_o}}\right)\!+\!(j_0-q_0)(\gamma +\log 2).
\end{align} 
By Proposition \ref{bahaviour-at-zero1} we arrive at the final result
\begin{align*}
Q=(j_0-q_0)(\gamma +\log   2)+ \log \left[\alpha_{j_oa_o}\rho (2\pi)^{-q/2}(-1)^{q_o-j_o}\right]- \\ - \log \textup{det}\left(\!\!\begin{array}{cc}\mathcal{A} & \mathcal{B} \\ \begin{array}{l}\kappa Id_{q_o} \qquad \qquad 0 \\ 0 \quad \textup{diag}(\kappa R^{\nu_l}+\nu_l R^{\nu_l-\frac{1}{2}})\end{array}& \begin{array}{l}(\kappa \log   R + \frac{1}{\sqrt{R}})Id_{q_o} \qquad 0 \\ 0 \quad \textup{diag}(\kappa R^{-\nu_l}-\nu_l R^{-\nu_l-\frac{1}{2}})\end{array}\end{array}\!\!\right).
\end{align*}
Exponentiating the expression proves the statement of the proposition.
\end{proof}
\begin{remark}\label{kernel-non-zero}
In case $\ker \mathcal{L}\neq \{0\}$ we can't apply Proposition \ref{general-result}. However the intermediate relation \eqref{result2} still holds. Further steps are possible if the asymptotics of $F(\mu)$ at zero is determined. 
\end{remark}

\subsection{Special Cases of Self-adjoint Extensions}
We compute the zeta-regularized determinants of some particular self-adjoint extensions of the model Laplacian $\triangle_{\nu-1/2}, \nu \geq 0$ in the notation of Subsection \ref{model-operator}, where we put $R=1$ for simplicity: 
$$\triangle_{\nu -1/2}= -\frac{d^2}{dx^2}+\frac{1}{x^2}\left[\nu^2-\frac{1}{4}\right]:C^{\infty}_0(0,1)\rightarrow C^{\infty}_0(0,1), \ \nu \geq 0.$$
According to Proposition \ref{laplacian-maximal} we get for the asymptotics of any $f \in \dom (\triangle_{\nu -1/2, \max})$ as $x\to 0$:
\begin{align} \label{laplacian-maxA} 
f(x)=&\, c_1(f) \cdot \sqrt{x} + c_2(f)\cdot \sqrt{x}\log (x) + O(x^{3/2}), \ \nu=0, \\ \label{laplacian-maxB}
f(x)=&\, c_1(f)\cdot x^{\nu+\frac{1}{2}} + c_2(f) \cdot x^{-\nu +\frac{1}{2}} + O(x^{3/2}), \ \nu \in (0,1), \\
f(x)=&\, O(x^{3/2}), \ \nu \geq 1.
\end{align}
The results of the previous subsection apply directly to the model situation for $\nu \in [0,1)$. In order to obtain results for $\nu \geq 1$ as well, we need to apply [L, Theorem 1.2]. We compute now a sequence of results on zeta-determinants for particular self-adjoint extensions which will become relevant afterwards.

\begin{cor}\label{1}
Let $D$ be the self-adjoint extension of $\triangle_{\nu-1/2}, \nu \geq 0$ with 
$$\dom(D):=\{f \in \dom (\triangle_{\nu -1/2,\max})|f(x)=O(\sqrt{x}), \ x\to 0;  f'(1)+\A f(1)=0\}.$$ 
Then for $\A\neq -\nu -1/2$ the operator $D$ is injective and 
$$\det\nolimits_{\zeta}(D)=\sqrt{2\pi}\frac{\A + \nu +1/2}{\Gamma(1+\nu)2^{\nu}}.$$
\end{cor}
\begin{proof} We first consider $\nu \in [0,1)$. In this case the extension $D$ amounts to $$\dom(D):=\{f \in \dom (\triangle_{\nu -1/2,\max})| c_2(f)=0,  f'(1)+\A f(1)=0\}.$$
Consider the polynomial $p(x,y)$ defined in Proposition \ref{asymptotics}. Its explicit form for the given extension is
\begin{align*}
&p(x,y)=-x, \ \textup{for $\nu =0$}, \\
&p(x,y)=-\frac{\Gamma (1+\nu)}{\Gamma(1-\nu)}2^{2\nu}y^{2\nu},  \ \textup{for $\nu \in (0,1)$}.
\end{align*}
Recall the definition of characteristic values $\A_0,j_0,a_{j_o\A_o}$ in Definition \ref{a-j-a}. For the given extension $D$ we obtain
\begin{align*}
&j_0=q_0=1, \ q=1, a_{j_o\A_o}=-1, \ \textup{for $\nu =0$}, \\
&j_0=q_0=0, \ q=1, a_{j_o\A_o}=-\frac{\Gamma (1+\nu)}{\Gamma(1-\nu)}2^{2\nu}, \ \textup{for $\nu \in (0,1)$}.
\end{align*}
Evaluating with Proposition \ref{bahaviour-at-zero1} the corresponding implicit eigenvalue function $F(\mu)$ at $\mu=0$ we obtain for any $\nu \in [0,1)$ $$F(0)=-\frac{1}{2}-\A-\nu.$$
Thus the condition $\A \neq -\nu -1/2$ guarantees $F(0)\neq 0$ and thus $\ker D=\{0\}$. Applying Proposition \ref{general-result} we obtain the desired formula.
\\[3mm] In order to obtain a result for all $\nu \geq 0$, we need to apply [L, Theorem 1.2]. Consider mappings $\phi, \psi:(0,1)\to \R$ such that 
\begin{align*}
\triangle_{\nu -1/2} \phi=0, \phi(x)=O(\sqrt{x}), x\to 0\ \textup{and} \ \phi(x)=x^{\nu+1/2}\phi_0(x), \phi_0(0)=1, \\
\triangle_{\nu -1/2} \psi=0, \psi'(1)+\A\psi(1)=0 \ \textup{and} \ \psi(1)=1.
\end{align*}
These solutions exist and are uniquely determined. In the sense of [L, (1.38a), (1.38b)] the solutions $\phi, \psi$ are "normalized" at $x=0, x=1$, respectively. In the present setup the normalized solution $\phi$ is given explicitly as follows $$\phi(x)=x^{\nu+1/2}.$$
In particular we obtain for the Wronski determinant $$W(\psi, \psi)=\phi'(1)\psi(1)-\phi(1)\psi'(1)=\A+\nu +1/2.$$
By assumption $\A\neq -\nu -1/2$ and hence $W(\psi, \phi)\neq 0$. Note from the fundamental system of solutions to $\triangle_{\nu-1/2}f=0$ in \eqref{fundamental1} and \eqref{fundamental2} that $\ker D=\{0\}$. Thus we can apply [L, Theorem 1.2] where in the notation therein $\nu_0=\nu$ and $\nu_1=-1/2$: $$\det\nolimits_{\zeta}D=\frac{\pi W(\psi, \phi)}{2^{\nu_0+\nu_1}\Gamma(1+\nu_0)\Gamma(1+\nu_1)}=\sqrt{2\pi}\frac{\A+\nu+1/2}{2^{\nu}\Gamma(1+\nu)}.$$
This proves the statement. 
\end{proof}

\begin{cor}\label{4}
Let $D$ be the self-adjoint extension of $\triangle_{\nu-1/2}, \nu \geq 0$ with 
$$\dom(D):=\{f \in \dom (\triangle_{\nu-1/2,\max})| f(x)=O(\sqrt{x}), x\to 0; f(1)=0\}.$$
The zeta-regularized determinant of this self-adjoint realization is given by
$$\det\nolimits_{\zeta}(D)=\frac{\sqrt{2\pi}}{\Gamma(1+\nu)2^{\nu}}.$$
\end{cor}
\begin{proof} We first consider $\nu \in [0,1)$. In this case the extension $D$ amounts to 
$$\dom(D):=\{f \in \dom (\triangle_{\nu-1/2,\max})| c_2(f)=0, f(1)=0\}.$$
As in the proof of Corollary \ref{1} we infer for the characteristic values of the extension $D$
\begin{align*}
&j_0=q_0=1, \ q=1, a_{j_o\A_o}=-1, \ \textup{for $\nu =0$}, \\
&j_0=q_0=0, \ q=1, a_{j_o\A_o}=-\frac{\Gamma (1+\nu)}{\Gamma(1-\nu)}2^{2\nu}, \ \textup{for $\nu \in (0,1)$}.
\end{align*}
Consider now the following self-adjoint extension of $\triangle_{\nu-1/2}$:
\begin{align*}
\dom (\triangle_{\nu-1/2}^{\mathcal{N}}):=\left\{\begin{array}{l} \{f \in \dom (\triangle_{\nu-1/2,\max})| c_2(f)=0, f(1)=0\}, \ \nu=0, \\ 
\{f \in \dom (\triangle_{\nu-1/2,\max})| c_1(f)=0, f(1)=0\}, \ \nu\in (0,1).
\end{array}\right.
\end{align*}
This extension is referred to as "Neumann-extension" (with Dirichlet boundary conditions at $x=1$) in [KLP1] and by [KLP1, Corollary 4.7] we have:
\begin{align*}
\det\nolimits_{\zeta}\triangle_{\nu-1/2}^{\mathcal{N}}=\sqrt{2\pi}\frac{2^{\nu}}{\Gamma(1-\nu)}, \ \nu\in [0,1).
\end{align*}
Note that for $\nu=0$ we have $\triangle_{\nu-1/2}^{\mathcal{N}}\equiv D$. Hence it remains to compute the zeta-determinant of $D$ for $\nu \in (0,1)$. Using [KLP1, Theorem 2.3], where arbitrary extensions at the cone-singularity are expressed in terms of the "Neumann extension", we obtain
\begin{align*}
\det\nolimits_{\zeta}(D)=\det\nolimits_{\zeta}\triangle_p^{\mathcal{N}}\cdot \frac{(-2e^{\gamma})^{q_o-j_o}}{a_{j_o\A_o}}\cdot \det \left(\begin{array}{cc}0&1\\1&1\end{array}\right)= \frac{\sqrt{2\pi}}{\Gamma(1+\nu)2^{\nu}}.
\end{align*}
This proves the statement for $\nu \in [0,1)$. In order to obtain the desired result for all $\nu\geq 0$, we apply [L, Theorem 1.2] by similar means as in the previous corollary.
\end{proof}\ \\
\\[-7mm] The remaining three results differ from the previous two corollaries by the fact that the self-adjoint realizations considered there are not injective. In this case we cannot apply [L, Theorem 1.2] and Proposition \ref{general-result}. We have to apply the intermediate result \eqref{result2}.

\begin{cor}\label{1'}
Consider a self-adjoint extension $D$ of $\triangle_{\nu-1/2}$ with $\nu=1/2$
$$\dom(D):=\{f \in \dom (\triangle_{\nu-1/2,\max})|c_2(f)=0, f'(1)- f(1)=0\}.$$
The zeta-regularized determinant of this self-adjoint realization is given by
$$\det\nolimits_{\zeta}(D)=\frac{2}{3}.$$
\end{cor}
\begin{proof}
As in the proof of Corollary \ref{1} we infer for the characteristic values of the extension $D$
$$j_0=q_0=0, \ q=1, a_{j_o\A_o}=-\frac{\Gamma (1+\nu)}{\Gamma(1-\nu)}2^{2\nu}=-1. $$
Evaluating with Proposition \ref{bahaviour-at-zero1} the corresponding implicit eigenvalue function $F(\mu)$ at $\mu=0$ we obtain as in Corollary \ref{1} $$F(0)=-\frac{1}{2}+1-\nu=0.$$
This implies $\ker D\neq\{0\}$ and so unfortunately we cannot apply Proposition \ref{general-result} directly. Note however from the definition of $F(\mu)$ in Proposition \ref{implicit-eigenvalue-equation}
\begin{align*}
F(\mu)=-\sqrt{\frac{\pi}{2}}\frac{1}{\sqrt{\mu}}\left[-\frac{1}{2}J_{1/2}(\mu)+\mu J'_{1/2}(\mu)\right].
\end{align*}
With $J_{1/2}(\mu)=\sqrt{\frac{2}{\pi \mu}}\sin (\mu)$ we compute further
\begin{align*}
F(\mu)=\frac{\sin (\mu)}{\mu}-\cos(\mu). 
\end{align*}
From the Taylor expansion of $\sin (\mu),\cos(\mu)$ around zero we get
$$F(\mu)=\frac{1}{3}\mu^2+O(\mu^4), \ \textup{as} \ |\mu|\to 0.$$
Thus we equivalently can consider a different implicit eigenvalue function $$\tilde{F}(\mu)=\frac{F(\mu)}{\mu^2}, \tilde{F}(0)=1/3.$$
By Remark \ref{kernel-non-zero} and the relation \eqref{result2} we obtain with $j_0-q_0=0, R=1$
\begin{align}\label{help1'}
\zeta'(0,D)=- \!\log  \left(\frac{\tilde{F}(i|t|)}{\tilde{C}}\right)- \frac{1}{\pi i}\int\limits_{\gamma_t}\!\log  \mu\frac{\tilde{F}'(\mu)}{\tilde{F}(\mu)}.
\end{align}
Note that the second summand is now an entire integral over a finite curve. Taking $|t|\to 0$ this integral vanishes. Recall that the constant $C$ was introduced in \eqref{G-F} to describe the constant term in the asymptotics of $\log F(ix)$. By construction we have a relation between $\tilde{C}$ associated to $\tilde{F}(\mu)$ and $C$ associated to the original $F(\mu)$ 
$$\tilde{C}=-C=-a_{j_o\A_o}\rho (2\pi)^{-q/2},$$
where $\rho$ is defined in the statement of Proposition \ref{asymptotics}, leading in the present situation to $\rho=\Gamma(1-\nu)2^{-\nu}=\sqrt{\pi /2}$. Hence $\tilde{C}$ computes explicitly to $\tilde{C}=1/2$. Inserting this now into \eqref{help1'} and taking $|t|\to 0$ we obtain
\begin{align*}
\zeta'(0,D)=- \!\log  \left(\frac{\tilde{F}(0)}{\tilde{C}}\right)=-\log \left(\frac{2}{3}\right).
\end{align*}
This proves the statement with $\det\nolimits_{\zeta}(D)=\exp(-\zeta'(0,D))$.
\end{proof}

\begin{cor}\label{2'}
Consider a self-adjoint extension $D$ of $\triangle_{\nu-1/2}$ with $\nu=1/2$
$$\dom(D):=\{f \in \dom (\triangle_{\nu-1/2,\max})|c_1(f)=0, f'(1)=0\}.$$
The zeta-regularized determinant of this self-adjoint realization is given by
$$\det\nolimits_{\zeta}(D)=2.$$
\end{cor}
\begin{proof}
Consider the polynomial $p(x,y)$ defined in Proposition \ref{asymptotics}. Its explicit form for the given extension is
$$p(x,y)=1.$$
Thus we get for the characteristic values $\A_0,j_0,a_{j_o\A_o}$, defined in Definition \ref{a-j-a} $$j_0=q_0=0, \ q=1, a_{j_o\A_o}=1. $$ Evaluating with Proposition \ref{bahaviour-at-zero1} the corresponding implicit eigenvalue function $F(\mu)$ at $\mu=0$ we obtain $$F(0)=\frac{1}{2}-\nu=0.$$
This implies $\ker D\neq\{0\}$ and so unfortunately we cannot apply Proposition \ref{general-result} directly. Note however from the definition of $F(\mu)$ in Proposition \ref{implicit-eigenvalue-equation}
\begin{align*}
F(\mu)=\sqrt{\frac{\pi}{2}}\sqrt{\mu}\left[\frac{1}{2}J_{-1/2}(\mu)+\mu J'_{-1/2}(\mu)\right].
\end{align*}
With $J_{-1/2}(\mu)=\sqrt{\frac{2}{\pi \mu}}\cos (\mu)$ we compute further
\begin{align*}
F(\mu)=-\mu\sin (\mu). 
\end{align*}
From the Taylor expansion of $\sin (\mu)$ around zero we get
$$F(\mu)=-\mu^2+O(\mu^4), \ \textup{as} \ |\mu|\to 0.$$
Thus we equivalently can consider a different implicit eigenvalue function $$\tilde{F}(\mu)=\frac{F(\mu)}{\mu^2}, \tilde{F}(0)=-1.$$
By Remark \ref{kernel-non-zero} and the relation \eqref{result2} we obtain with $j_0-q_0=0, R=1$
\begin{align}\label{help2'}
\zeta'(0,D)=- \!\log  \left(\frac{\tilde{F}(i|t|)}{\tilde{C}}\right)- \frac{1}{\pi i}\int\limits_{\gamma_t}\!\log  \mu\frac{\tilde{F}'(\mu)}{\tilde{F}(\mu)}.
\end{align}
Note that the second summand is now an entire integral over a finite curve. Taking $|t|\to 0$ this integral vanishes. As in the previous corollary we find by construction a relation between $\tilde{C}$ associated to $\tilde{F}(\mu)$ and $C$ associated to the original $F(\mu)$ 
$$\tilde{C}=-C=-a_{j_o\A_o}\rho (2\pi)^{-q/2},$$
where $\rho=\Gamma(1-\nu)2^{-\nu}=\sqrt{\pi /2}$ is defined in the statement of Proposition \ref{asymptotics}. The constant $\tilde{C}$ computes explicitly to $\tilde{C}=-1/2$. Inserting this now into \eqref{help2'} and taking $|t|\to 0$ we obtain
\begin{align*}
\zeta'(0,D)=- \!\log  \left(\frac{\tilde{F}(0)}{\tilde{C}}\right)=-\log 2.
\end{align*}
This proves the statement with $\det\nolimits_{\zeta}(D)=\exp(-\zeta'(0,D))$.
\end{proof}

\begin{cor}\label{2}
Consider a self-adjoint extension $D$ of $\triangle_{\nu-1/2}$ with $\nu=0$
$$\dom(D):=\{f \in \dom (\triangle_{\nu-1/2,\max})|c_2(f)=0, f'(1)-\frac{1}{2} f(1)=0\}.$$
The zeta-regularized determinant of this self-adjoint realization is given by
$$\det\nolimits_{\zeta}(D)=\sqrt{\frac{\pi}{2}}.$$
\end{cor}
\begin{proof}
Consider the polynomial $p(x,y)$ defined in Proposition \ref{asymptotics}. Its explicit form for the given extension is
$$p(x,y)=-x.$$
According to Definition \ref{a-j-a} we obtain from above the characteristic values of the extension $D$
$$j_0=q_0=1, \ q=1, a_{j_o\A_o}=-1. $$ Evaluating with Proposition \ref{bahaviour-at-zero1} the corresponding implicit eigenvalue function $F(\mu)$ at $\mu=0$ we obtain $$F(0)=\det\left(\begin{array}{cc}0&1\\0&1\end{array}\right)=0.$$
This implies $\ker D\neq\{0\}$ and so unfortunately we cannot apply Proposition \ref{general-result} directly. Note however from the definition of $F(\mu)$ in Proposition \ref{implicit-eigenvalue-equation}
\begin{align*}
F(\mu)=-\mu J'_0(\mu)= \mu J_1(\mu),
\end{align*}
where we used the identity $J'_0=J_{-1}=-J_1$. Hence as $|\mu|\to 0$ we obtain the following asymptotics
$$F(\mu)=\mu J_1(\mu)\sim \frac{\mu^2}{2\Gamma(2)}=\frac{\mu^2}{2}.$$
Thus we equivalently can consider a different implicit eigenvalue function $$\tilde{F}(\mu)=\frac{F(\mu)}{\mu^2}, \tilde{F}(0)=1/2.$$
By Remark \ref{kernel-non-zero} and the relation \eqref{result2} we obtain with $j_0-q_0=0, R=1$
\begin{align}\label{help2}
\zeta'(0,D)=- \!\log  \left(\frac{\tilde{F}(i|t|)}{\tilde{C}}\right)- \frac{1}{\pi i}\int\limits_{\gamma_t}\!\log  \mu\frac{\tilde{F}'(\mu)}{\tilde{F}(\mu)}.
\end{align}
Note that the second summand is now an entire integral over a finite curve. Taking $|t|\to 0$ this integral vanishes. As before we find by construction a relation between $\tilde{C}$ associated to $\tilde{F}(\mu)$ and $C$ associated to the original $F(\mu)$ 
$$\tilde{C}=-C=-a_{j_o\A_o}\rho (2\pi)^{-q/2},$$
where $\rho$ is defined in the statement of Proposition \ref{asymptotics} and equals $1$ in the present case. The constant $\tilde{C}$ computes explicitly to $\tilde{C}=1/\sqrt{2\pi}$. Inserting this now into \eqref{help2} and taking $|t|\to 0$ we obtain
\begin{align*}
\zeta'(0,D)=- \!\log  \left(\frac{\tilde{F}(0)}{\tilde{C}}\right)=-\log \left(\sqrt{\frac{\pi}{2}}\right).
\end{align*}
This proves the statement with $\det\nolimits_{\zeta}(D)=\exp(-\zeta'(0,D))$.
\end{proof}

\section{Boundary Conditions for the de Rham Laplacian on a Bounded Generalized Cone}\label{section-laplace}\
\\[-3mm] In this section we discuss the geometry of a bounded generalized cone and following [L3] decompose the associated de Rham Laplacian in a compatible way with respect to its relative self adjoint extension. This decomposition allows us to study the relative boundary conditions of the Laplace operator explicitly and provides a basis for the computation of the associated zeta-regularized determinants.
\\[3mm] The question about the self-adjoint extensions of the Laplacians on differential forms of a fixed degree, on manifolds with conical singularities is addressed in [BL2, Theorems 3.7 and 3.8]. There, among many other issues, the relative extension is shown to coincide with the Friedrich's extension at the cone singularity, outside of the middle degrees.
\\[3mm] Using the decomposition of the complex we obtain further explicit results without the degree limitations.

\subsection{Regular-Singular Operators}\label{rescaling-trick}
Consider a bounded generalized cone $M=(0,R]\times N$ over a closed oriented Riemannian manifold $(N,g^N)$ of dimension $\dim N =n$, with the Riemannian metric on $M$ given by a warped product $$g^M = dx^2 \oplus x^2g^N.$$ 
The volume forms on $M$ and $N$, associated to the Riemannian metrics $g^M$ and $g^N$, are related as follows: 
$$\textup{vol}(g^M)=x^n dx \wedge \textup{vol}(g^N).$$
Consider as in [BS, (5.2)] the following separation of variables map, which is linear and bijective:
\begin{align}\label{separation}
\Psi_k : C^{\infty}_0((0,R),\Omega^{k-1}(N)\oplus \Omega^k(N))\to \Omega_0^k(M) \\
(\phi_{k-1},\phi_k)\mapsto x^{k-1-n/2}\phi_{k-1}\wedge dx + x^{k-n/2}\phi_k, \nonumber
\end{align}
where $\phi_k,\phi_{k-1}$ are identified with their pullback to $M$ under the natural projection $\pi: (0,R]\times N\to N$ onto the second factor, and $x$ is the canonical coordinate on $(0,R]$. Here $\Omega_0^k(M)$ denotes differential forms of degree $k=0,..,n+1$ with compact support in the interior of $M$. The separation of variables map $\Psi_k$ extends to an isometry with respect to the $L^2$-scalar products, induced by the volume forms vol$(g^M)$ and vol$(g^N)$. 
\begin{prop}\label{unitary} The separation of variables map \eqref{separation} extends to an isometric identification of $L^2-$Hilbert spaces
\begin{align*}
\Psi_k: L^2([0,R], L^2(\wedge^{k-1}T^*N\oplus \wedge^kT^*N, \textup{vol}(g^N)), dx)\to L^2(\wedge^kT^*M, \textup{vol}(g^M)).
\end{align*}
\end{prop}\ \\
\\[-7mm] Under this identification we obtain for the exterior derivative, as in [BS, (5.5)]
\begin{equation}\label{derivative} 
\Psi_{k+1}^{-1} d_k \Psi_k= \left( \begin{array}{cc}0&(-1)^k\partial_x\\0&0\end{array}\right)+\frac{1}{x}\left( \begin{array}{cc}d_{k-1,N}&c_k\\0&d_{k,N}\end{array}\right),
\end{equation}
where $c_k=(-1)^k(k-n/2)$ and $d_{k,N}$ denotes the exterior derivative on differential forms over $N$ of degree $k$. Taking adjoints we find 
\begin{equation}\label{coderivative} 
\Psi_k^{-1} d_k^t \Psi_{k+1}= \left( \begin{array}{cc}0&0\\(-1)^{k+1}\partial_x&0\end{array}\right)+\frac{1}{x}\left( \begin{array}{cc}d_{k-1,N}^t&0\\c_k&d_{k,N}^t\end{array}\right).
\end{equation}
\\[3mm] Consider now the Gauss-Bonnet operator $D_{GB}^+$ mapping forms of even degree to forms of odd degree. The Gauss-Bonnet operator acting on forms of odd degree is simply the formal adjoint $D_{GB}^-=(D_{GB}^+)^t$. With respect to $\Psi_{+}:=\oplus \Psi_{2k}$ and $\Psi_{-}:=\oplus \Psi_{2k+1}$ the relevant operators take the following form:
\begin{align}
&\Psi^{-1}_-D_{GB}^+\Psi_+ = \frac{d}{dx}+\frac{1}{x}S_0, \quad \Psi^{-1}_+D_{GB}^-\Psi_- = -\frac{d}{dx}+\frac{1}{x}S_0, \label{gauss-bonnet}\\
 &\Psi^{-1}_+\triangle^+ \Psi_+= \Psi^{-1}_+(D_{GB}^+)^t\Psi_-\Psi_-^{-1}D_{GB}^+ \Psi_+ = -\frac{d^2}{dx^2}+ \frac{1}{x^2}S_0(S_0+1),\label{laplacian} \\ \nonumber
  &\Psi^{-1}_-\triangle^- \Psi_-= \Psi^{-1}_-(D_{GB}^-)^t\Psi_+\Psi_+^{-1}D_{GB}^- \Psi_- = -\frac{d^2}{dx^2}+ \frac{1}{x^2}S_0(S_0-1).
\end{align}
where $S_0$ is a first order elliptic differential operator on $\Omega^*(N)$. The transformed Gauss-Bonnet operators in \eqref{gauss-bonnet} are regular singular in the sense of [BS] and [Br, Section 3]. Moreover, the Laplace Operator on $k$-forms over $M$ transforms to 
\begin{align}\label{A_k,S_0}
\Psi_k \triangle_k \Psi_k^{-1}=-\frac{d^2}{dx^2}+ \frac{1}{x^2}A_k,
\end{align}
where the operator $A_k$ is simply the restriction of $S_0(S_0+(-1)^k)$ to $\Omega^{k-1}(N)\oplus \Omega^k(N)$. Note, that under the isometric identification $\Psi_*$ the previous non-product situation of the bounded generalized cone $M$ is now incorporated in the $x$-dependence of the tangential parts of the geometric Gauss-Bonnet and Laplace operators.
\\[3mm] Next we take boundary conditions into account and consider their behaviour under the isometric identification $\Psi_*$. More precisely consider the exterior derivatives and their formal adjoints on differential forms with compact support in the interior of $M$:
\begin{align*}
d_k:\Omega^k_0(M)\rightarrow \Omega_0^{k+1}(M), \\
d^t_k: \Omega^{k+1}_0(M)\rightarrow \Omega^k_0(M).
\end{align*}
Define the minimal closed extensions $d_{k,\min}$ and $d^t_{k,\min}$ as the graph closures in $L^2(\bigwedge\nolimits^*T^*M,\textup{vol}(g^M))$ of the differential operators $d_k$ and $d^t_k$ respectively. 
\\[3mm] The operators $d_{k,\min}$ and $d^t_{k,\min}$ are closed and densely defined. In particular we can form the adjoint operators and set for the maximal extensions:
\begin{align*}
d_{k,\max}:=(d^t_{k,\min})^*, \quad d^t_{k,\max}:=(d_{k,\min})^*.
\end{align*}
These definitions correspond to Theorem \ref{max-min-theorem}. The following result is an easy consequence of the definitions of the minimal and maximal extensions and of Proposition \ref{unitary}.
\begin{prop}\label{isometry-bc}
\begin{align*}
\Psi^{-1}_k(\dom (d_{k,\min}))=\dom ([\Psi^{-1}_{k+1}d_k\Psi_k]_{\min}), \\
\Psi^{-1}_k(\dom (d_{k,\max}))=\dom ([\Psi^{-1}_{k+1}d_k\Psi_k]_{\max}).
\end{align*}
\end{prop}\ \\
\\[-7mm] Similar statements hold for the minimal and maximal extensions of the formal adjoint operators $d^t_k$. The minimal and the maximal extensions of the exterior derivative give rise to self-adjoint extensions of the associated Laplace operator
$$\triangle_k=d^t_kd_k+d_{k-1}d_{k-1}^t.$$ 
It is important to note that there are self-adjoint extensions of $\triangle_k$ which do not come from closed extensions of $d_k$ and $d_{k-1}$, compare the notion of "ideal boundary conditions" in [BL1]. However the most relevant self-adjoint extensions of the Laplacian indeed seem to come from closed extensions of the exterior derivatives. 
\\[3mm] We are interested in the relative and the absolute self-adjoint extensions of $\triangle_k$, defined as follows:
\begin{align}\label{definition}
\hspace{25mm}\triangle_k^{rel}:= &\, d^*_{k,\min}d_{k,\min}+d_{k-1,\min}d^*_{k-1,\min}= \\
\hspace{25mm}=&\, d^t_{k,\max}d_{k,\min}+d_{k-1,\min}d^t_{k-1,\max}, \nonumber \\
\triangle_k^{abs}:= &\, d^*_{k,\max}d_{k,\max}+d_{k-1,\max}d^*_{k-1,\max}= \\
=&\, d^t_{k,\min}d_{k,\max}+d_{k-1,\max}d^t_{k-1,\min}. \nonumber
\end{align}
As a direct consequence of the previous proposition and Proposition \ref{unitary} we obtain for the relative self-adjoint extension (absolute self-adjoint extension is discussed in a similar way):
\begin{cor}\label{unitary-bc}
Consider the following two complexes
\begin{align*}
(\Omega^*_0(M),d_k), \quad (C^{\infty}_0((0,R),C^{\infty}(\wedge^{k-1}T^*N\oplus \wedge^kT^*N)), \widetilde{d}_k:=\Psi^{-1}_{k+1}d_k\Psi_k).
\end{align*}
Then the relative self-adjoint extensions of the associated Laplacians 
\begin{align*}
\triangle_k^{rel}&= d^*_{k,\min}d_{k,\min}+d_{k-1,\min}d^*_{k-1,\min}, \\
\widetilde{\triangle}_k^{rel}&= \widetilde{d}^*_{k,\min}\widetilde{d}_{k,\min}+\widetilde{d}_{k-1,\min}\widetilde{d}^*_{k-1,\min}
\end{align*}
are spectrally equivalent, with $\Psi^{-1}_k(\dom (\triangle^{rel}_k))=\dom (\widetilde{\triangle}^{rel}_k)$ and $$\widetilde{\triangle}^{rel}_k=\Psi_k^{-1}\triangle^{rel}_k \Psi_k.$$
\end{cor}\ \\
\\[-7mm] As a consequence of Corollary \ref{unitary-bc} we can deal with the minimal extension of the unitarily transformed exterior differential $\Psi^{-1}_{k+1}d_k\Psi_k$ and the relative extension of the unitarily transformed Laplacian $\Psi^{-1}_{k}\triangle_k\Psi_k$ without loss of generality. By a small abuse of notation we denote the operators again by $d_{k,\min}$ and $\triangle^{rel}_k$, in order to keep the notation simple. 

\subsection{Decomposition of the de Rham Laplacian}\label{relative-singular}
Our goal is the explicit determination of the domain of $\triangle_k^{rel}, k=0,..,m=\dim M$.  We restrict ourselves to the relative extension, since the absolute extension is treated analogously. 
\\[3mm] By the convenient structure \eqref{A_k,S_0} of the Laplacian $\triangle_k$ one is tempted to write 
$$\triangle_k=\bigoplus_{\lambda \in \textup{Sp}(A_k)} -\frac{d^2}{dx^2}+\frac{\lambda}{x^2},$$ 
and study the boundary conditions induced on each one-dimensional component. However this decomposition might be incompatible with the boundary conditions, so the discussion of the corresponding self-adjoint realization might not reduce to simple one-dimensional problems. This is in fact the case for the relative boundary conditions, which (by definition) determine the domain of the relative extension $\triangle_k^{rel}$. 
\\[3mm] Nevertheless we infer from the decomposition above and \eqref{H-2-loc} the regularity properties for elements $\phi\in \dom (\triangle_{k,\max})$, needed in the formulation of Proposition \ref{relative-regular} below.
\\[3mm] At the cone face $\{x=R\}\times N$ the relative boundary conditions are derived from the following trace theorem of L. Paquet:
\begin{thm}\label{trace-theorem} \textup{[P, Theorem 1.9]} Let $K$ be a compact oriented Riemannian manifold with boundary $\partial K$ and let $\iota: \partial K \hookrightarrow K$ be the natural inclusion. Then the pullback $\iota^*:\Omega^k(K) \to \Omega^k(\partial K)$ with $\Omega^k(\partial K)=\{0\}$ for $k=\dim K$, extends continuously to the following linear surjective map: $$\iota^*:\dom (d_{k,\max})\rightarrow \dom (d_{k,\partial K}^{-1/2}),$$ where $d_{k,\partial K}^{-1/2}$ is the closure of the exterior derivative on $\partial K$ in the Sobolev space $H^{-1/2}(\wedge^*T^*\partial K)$ and $d_{k, \max}$ the maximal extension of the exterior derivative on $K$. The domains $\dom (d_{k,\max})$ and $\dom (d_{k,\partial K}^{-1/2})$ are Hilbert spaces with respect to the graph-norms of the corresponding operators.
\end{thm}

\begin{prop}\label{relative-regular} Let $\gamma\in C^{\infty}[0,R]$ be a smooth cut-off function, vanishing identically at $x=0$ and being identically one at $x=R$. Then
\begin{align*}
\gamma \dom (\triangle_k^{rel}) = \{\Psi_k(\phi_{k-1},\phi_k) \in \gamma \dom(\triangle_{k,\max})| \phi_k(R)=0, \\  
\phi'_{k-1}(R)- \frac{(k-1-n/2)}{R}\phi_{k-1}(R)=0\}.
\end{align*}
\end{prop}
\begin{proof} Let $r\in (0,R)$ be fixed and consider the associated natural inclusions
\begin{align*}
\chi&: [0,R]\times N =:M_r \hookrightarrow M, \\
\iota&: \{R\}\times N \equiv N \hookrightarrow M, \\
\iota_r&: \{R\}\times N \equiv N \hookrightarrow M_r. 
\end{align*}
We obviously have $\iota =\chi \circ \iota_r$. The inclusions above induce pullbacks of differential forms. The pullback map $\chi^*:\Omega^k(M) \to \Omega^k(M_r)$ is simply a restriction and extends to a continuous linear map $$\chi^*:\dom (d_{k,\max})\to \dom (d^r_{k,\max}),$$ where $d^r_k$ is the $k-$th exterior derivative on $M_r\subset M$ and the domains are endowed with the graph norms of the corresponding operators. Applying Theorem \ref{trace-theorem} to the compact manifold $M_r$, we deduce that $\iota^*=\iota_r^*\circ \chi^*$ extends to a continuous linear map 
\begin{align}\label{sternchen}
\iota^*:\dom (d_{k,\max})\to \dom (d^{-1/2}_{k,N}).
\end{align}
Now, continuity of $\iota^*$ together with the definition of the minimal domain $\dom (d_{k,\min})$ implies $$\gamma \dom (d_{k,\min})\subseteq \{\phi\in \gamma \dom (d_{k,\max})|\iota^*\phi=0\}.$$
Equality in the relation above follows from the Lagrange identity for $d_k$. We obtain for the relative boundary conditions at the cone base: $$\gamma \dom (\triangle^{rel}_k)=\{\phi \in \gamma \dom (\triangle_{k,\max})| \iota^*\phi=0, \iota^*(d^t_{k-1}\phi)=0\}.$$
Now the statement of the proposition follows from the explicit action of $d^t_{k-1}$ under the isometric identification $\Psi_*$ and the fact that for $\Psi_k(\phi_{k-1},\phi_k)\in \dom (\triangle_{k,\max})$ we have $\iota^*(\Psi_k(\phi_{k-1},\phi_k))= R^{k-n/2}\phi_k(R)$. 
\end{proof}\ \\
\\[-7mm] In order to identify the relative boundary conditions at the cone singularity, we decompose $\triangle_k$ into a direct sum of operators such that the decomposition is compatible with the relative self-adjoint extension.
\\[3mm] Compatibility of a decomposition means explicitly the following in the context of our presentation.
\begin{defn}\label{def-compatibility}
Let $D$ be a closed operator in a Hilbert space $H$. Let $H_1$ be a closed subspace of $H$ and $H_2:=H_1^{\perp}$. We say the decomposition $H=H_1\oplus H_2$ is compatible with $D$ if $D(H_j\cap \dom (D))\subset H_j, j=1,2$ and for any $\phi_1\oplus \phi_2\in \dom (D)$ we get $\phi_1, \phi_2 \in \dom (D)$.
\end{defn}\ \\
\\[-7mm] This definition corresponds to [W2, Exercise 5.39] where the subspaces $H_j,j=1,2$ are called the "reducing subspaces of $D$". We have the following result:
\begin{prop}\label{lemma-compatibility}\textup{[W2, Theorem 7.28]}
Let $D$ be a self-adjoint operator in a Hilbert space $H$. Let $H_1$ be a closed subspace of $H$ and $H_2:=H_1^{\perp}$. Let the decomposition $H=H_1\oplus H_2$ be compatible with $D$. Then each operator $D_i:=D|H_i,i=1,2$ with domain 
$$\dom (D_i):=\dom (D)\cap H_i, \ i=1,2$$
is a self-adjoint operator in $H_i$. In other words, the induced decomposition $D=D_1\oplus D_2$ is an orthogonal decomposition of $D$ into sum of two self-adjoint operators.
\end{prop}
\begin{defn}\label{induced}
In the setup of Proposition \ref{lemma-compatibility} we say $D_i, i=1,2$ is a self-adjoint operator \textup{"induced"} by $D$.
\end{defn} \ \\
\\[-7mm] In order to simplify notation, put: 
\begin{align*}
H^k:&=L^2([0,R], L^2(\wedge^{k-1}T^*N\oplus \wedge^{k}T^*N, \textup{vol}(g^N)), dx), \\ H^*:&=\bigoplus_{k\geq 0}H^k,
\end{align*}
where $H^k$ are mutually orthogonal in $H^*$. The following result follows straightforwardly by the definition of $\triangle_k^{rel}$ and gives a practical condition for compatibility of a decomposition of $H^k$ with the self-adjoint realization $\triangle_k^{rel}$.

\begin{prop}\label{rel-compatibility}
Let $H^k=H_1\oplus H_2, H_2:=H_1^{\perp}$ be an orthogonal decomposition into closed subspaces, such that $\triangle^{rel}_k(H_j\cap\dom(\triangle^{rel}_k))\subset H_j, j=1,2$. Assume that for $D\in \{d_k, d_k^t,d_k^td_k, d_{k-1}d_{k-1}^t\}$ the images $D_{\max}(H_j\cap \dom (D_{\max})), j=1,2$ are mutually orthogonal in $H^*$. Then the decomposition $H^k=H_1\oplus H_2$ is compatible with the relative extension $\triangle_k^{rel}$.
\end{prop}\ \\
\\[-7mm] Now we can present, following [L3], a decomposition of $H^*$, compatible with $\triangle_*^{rel}$. In fact this decomposition is one of the essential aspects in the computation of analytic torsion of a bounded generalized cone, in [BV2]. To describe the decomposition in convenient terms, we denote by $\triangle_{k,ccl,N}$ the Laplace operator on coclosed $k-$forms on $N$ and introduce some notation
\begin{align*}
&V_k:=\{\la \in \textup{Spec}\triangle_{k,ccl,N}\}\backslash \{0\}, \\
&E_{\lambda}^k:=\{\w \in \Omega^k(N)|\triangle_{k,N}\w=\lambda \w , d^t_N\w=0\}, \\
&\widetilde{E}_{\la}^k:=E_{\la}^k\oplus d_N E_{\lambda}^k, \quad \mathcal{H}^k(N):=E_0^k.
\end{align*}
Here $k=0,..,\dim N=n$ and the eigenvalues of $\triangle_{k,ccl,N}$ in $V_k$ are counted with their multiplicities, so that each single $E_{\la}^k$ is a one-dimensional subspace. The eigenvectors for a $\la \in V_k, k=0,..,n$ with multiplicity bigger than $1$ are chosen to be mutually orthogonal with respect to the $L^2$-inner product on $N$. Further let for each $\mathcal{H}^k(N)$ choose an orthonormal basis of harmonic forms $\{u^k_i\}$ with $i=1,..,\dim \mathcal{H}^k(N)$. 
\\[3mm] Then by the Hodge decomposition on $N$ we obtain for any fixed degree $k=0,..,n+1$ (put $\Omega^{n+1}(N)=\Omega^{-1}(N)=\{0\}$) 
\begin{align}\label{decomposition1}
\Omega^{k-1}(N)\oplus \Omega^k(N)=\left[\bigoplus\limits_{i=1}^{\dim \mathcal{H}^{k-1}(N)} \langle u^{k-1}_i\rangle \right]\oplus \left[\bigoplus\limits_{i=1}^{\dim \mathcal{H}^{k}(N)} \langle u^{k}_i \rangle \right]\\ \oplus \left[\bigoplus\limits_{\la \in V_{k-1}}\widetilde{E}_{\la}^{k-1}\right] \oplus \left[\bigoplus\limits_{\la \in V_{k-2}}d_NE_{\la}^{k-2}\right]\oplus \left[\bigoplus\limits_{\la \in V_{k}}E_{\la}^{k}\right].\nonumber
\end{align}
\begin{thm}\label{decomposition2}
The decomposition \eqref{decomposition1} induces an orthogonal decomposition of $H^k$, compatible with the relative extension $\triangle_k^{rel}$.
\end{thm}
\begin{proof}
The decomposition of $H^k$ induced by \eqref{decomposition1} is orthogonal, since the decomposition \eqref{decomposition1} is orthogonal with respect to the $L^2$-inner product on $N$. Applying now $d_k,d^t_kd_k$ and $d_{k-1}, d_{k-1}d_{k-1}^t$ to each of the orthogonal components we find that the images remain mutually orthogonal, so we obtain with Proposition \ref{rel-compatibility} the desired statement.
\end{proof}

\subsection{The Relative Boundary Conditions}\label{relative-singularII}
By Proposition \ref{lemma-compatibility} the orthogonal decomposition of $H^k$ in Theorem \ref{decomposition2} corresponds to a decomposition of $\triangle_k^{rel}$ into an orthogonal sum of self-adjoint operators. This decomposition is discussed by M. Lesch in [L3]. Using the decomposition we can now determine explicitly the boundary conditions for each of the self-adjoint components, up to the self-adjoint extension induced by $\triangle_k^{rel}$ over $$L^2((0,R),\widetilde{E}^{k-1}_{\lambda}), \lambda \in V_{k-1}.$$ We do not determine the boundary conditions for these particular self-adjoint components. However even for these components we can reduce the zeta-determinant calculations, which we perform in the next section, to other well-understood problems.
\\[3mm] Let $\psi \in E_{\la}^k, \la \in V_k, k=1,..,n$ be a fixed non-zero generator of $E_{\la}^k$. Put
\begin{align*}
&\xi_1:=(0,\psi)\in \Omega^{k-1}(N)\oplus \Omega^{k}(N), \\ &\xi_2:=(\psi,0)\in \Omega^{k}(N)\oplus \Omega^{k+1}(N), \\
&\xi_3:=(0,\frac{1}{\sqrt{\lambda}}d_N\psi)\in \Omega^{k}(N)\oplus \Omega^{k+1}(N), \\ &\xi_4:=(\frac{1}{\sqrt{\lambda}}d_N\psi,0)\in \Omega^{k+1}(N)\oplus \Omega^{k+2}(N).
\end{align*}
Then $C^{\infty}_0((0,R),\langle \xi_1,\xi_2,\xi_3,\xi_4\rangle)$ is invariant under $d,d^t$ and we obtain a subcomplex of the de Rham complex:
\begin{align}\label{subcomplex}
0 \rightarrow C_0^{\infty}((0,R),\left< \xi_1\right>) \xrightarrow{d_0} C_0^{\infty}((0,R),\left<\xi_2,\xi_3\right>) \xrightarrow{d_1}C_0^{\infty}((0,R),\left<\xi_4\right>)\rightarrow 0,
\end{align}
where $d_0,d_1$ take the following form with respect to the chosen basis:
\begin{align*}
d_0^{\psi}=\left(\begin{array}{c}(-1)^k\partial_x+\frac{c_k}{x}\\ x^{-1}\sqrt{\la}\end{array}\right), \quad d_1^{\psi}=\left(x^{-1}\sqrt{\la}, \ (-1)^{k+1}\partial_x+\frac{c_{k+1}}{x}\right).
\end{align*}
By Proposition \ref{lemma-compatibility} and Theorem \ref{decomposition2} we obtain for the induced (in the sense of Definition \ref{induced}) self-adjoint extensions:
\begin{align}
\dom (\triangle_k^{rel})\cap L^2((0,R),E_{\la}^k)=\dom (d_{0,\max}^td_{0,\min})=:\dom(\triangle_{0,\lambda}^k), \\
\dom (\triangle_{k+2}^{rel})\cap L^2((0,R),d_N E_{\la}^k)=\dom (d_{1,\min}d^t_{1,\max})=:\dom(\triangle_{2,\lambda}^k), \\
\dom (\triangle_{k+1}^{rel})\cap L^2((0,R),\widetilde{E}_{\la}^k)=\dom (d_{0,\min}d^t_{0,\max}+d_{1,\max}^td_{1,\min})=:\dom(\triangle^k_{\lambda}).
\end{align}
Note further that $d_0^td_0$ and $d_1d_1^t$ both act as the following regular-singular model Laplacian 
\begin{align}\label{lapl1}
\triangle:=-\frac{d^2}{dx^2}+\frac{1}{x^2}\left(\la +\left[k+\frac{1}{2}-\frac{n}{2}\right]^2-\frac{1}{4}\right),
\end{align}
under the identification of any $\phi =f\cdot \xi_{i}\in C^{\infty}_0((0,R),\langle \xi_{i}\rangle),i=1,4$ with its scalar part $f \in C^{\infty}_0(0,R)$. We continue under this identification from here on, as asserted in the next remark.
\begin{remark}\label{scalar-identification}
Let $V=\langle v \rangle$ denote any one-dimensional Hilbert space. Consider a particular type of a differential operator 
\begin{align*}
\mathcal{P}:C^{\infty}_0((0,R),V)&\rightarrow C^{\infty}_0((0,R),V), \\
f\cdot v &\mapsto (Pf)\cdot v,
\end{align*} 
where $f\in C^{\infty}_0(0,R)$ and $P$ is a scalar differential operator on $C^{\infty}_0(0,R)$. We call $f$ and $P$ the "\textup{scalar parts}" of $f\cdot v$ and $\mathcal{P}$, respectively. 
\\[3mm] We can reduce without loss of generality the spectral analysis of self-adjoint extensions of $\mathcal{P}$ to the spectral analysis of self-adjoint extensions of $P$ by identifying the $V-$valued functions $f\cdot v$ with their scalar parts. We fix this identification henceforth.
\end{remark}\ \\
\\[-7mm] In view of Corollary \ref{ess.s.a} we have to distinguish two cases. The first case is 
\begin{align}\label{lcc}
\la +\left[k+\frac{1}{2}-\frac{n}{2}\right]^2-\frac{1}{4} < \frac{3}{4},
\end{align}
so that $\triangle$ is in the limit circle case at $x=0$. Hence (note $\lambda \in V_k$ and so $\lambda \gneqq 0$) $$p:=\sqrt{\la +\left[k+\frac{1}{2}-\frac{n}{2}\right]^2}-\frac{1}{2}\in \left(-\frac{1}{2},\frac{1}{2}\right).$$ Then we get $\triangle=\triangle_p$ in the notation of Subsection \ref{model-operator}. Since $p \in (-1/2,1/2)$ we obtain from Proposition \ref{laplacian-maximal} for the asymptotics of elements $f \in \dom (\triangle_{\max})$:
\begin{align}\label{lcc2}
f(x)=c_1(f)\cdot x^{p+1} + c_2(f) \cdot x^{-p} + O(x^{3/2}).
\end{align}
In the second case 
\begin{align}\label{lpc!}
\la +\left[k+\frac{1}{2}-\frac{n}{2}\right]^2-\frac{1}{4} \geq \frac{3}{4}
\end{align}
the Laplacian $\triangle$ is by Corollary \ref{ess.s.a} in the limit point case at $x=0$ and hence in this case boundary conditions at zero are redundant. We can now compute the domains $\dom (\triangle^k_{0,\lambda})$ and $\dom (\triangle^k_{2,\lambda})$ explicitly.
\begin{lemma}\label{relative0}
Identify any $\phi \in \dom (\triangle_{0,\lambda}^k)$ with its scalar part $f \in \dom (\triangle_{p,\max})$. Under this identification we obtain for $\triangle_p$ in the limit circle case \eqref{lcc} at $x=0$, in the notation of \eqref{lcc2}
$$\dom(\triangle_{0,\lambda}^k)=\{f \in \dom (\triangle_{p,\max}) | c_2(f)=0, f(R)=0\}.$$
In the limit point case \eqref{lpc!} at $x=0$, we obtain
$$\dom(\triangle_{0,\lambda}^k)=\{f \in \dom (\triangle_{p,\max}) | f(R)=0\}.$$
\end{lemma}
\begin{proof}
Consider $\phi \in \dom (\triangle_{0,\lambda}^k)$ with its scalar part $f \in \dom (\triangle_{p,\max})$. We begin with the limit circle case at $x=0$. Since $\phi \in \dom (d_{0,\min})\subset \dom(d_{0,\max})$ we deduce from the explicit form of $d_0$ that $f \in \dom (1/x)_{\max}$, where $1/x$ is the obvious multiplication operator. Since with $p \in (-1/2,1/2)$ $$\textup{id}^{-p}\notin \dom (1/x)_{\max}$$ we deduce $c_2(f)=0$. On the other hand we infer from Proposition \ref{relative-regular} $$f(R)=0.$$ This proves the inclusion $\subset$ in the first statement. Since both sides of the inclusion define self-adjoint extensions and these are maximally symmetric, the inclusion must be an equality.
\\[3mm] For the limit point case at $x=0$ the argumentation is similar up to the fact that the boundary conditions at $x=0$ are redundant by Corollary \ref{ess.s.a}.
\end{proof}
\begin{lemma}\label{relative2}
Identify any $\phi \in \dom (\triangle_{2,\lambda}^k)$ with its scalar part $f \in \dom (\triangle_{p,\max})$. Under this identification we obtain for $\triangle_p$ in the limit circle case \eqref{lcc} at $x=0$, in the notation of \eqref{lcc2}
$$\dom(\triangle_{2,\lambda}^k)=\{f \in \dom (\triangle_{p,\max}) | c_2(f)=0, f'(R)-\frac{(k+1-n/2)}{R}f(R)=0\}.$$
In the limit point case \eqref{lpc!} at $x=0$, we obtain
$$\dom(\triangle_{2,\lambda}^k)=\{f \in \dom (\triangle_{p,\max}) | f'(R)-\frac{(k+1-n/2)}{R}f(R)=0\}.$$
\end{lemma}
\begin{proof}
Consider $\phi \in \dom (\triangle_{2,\lambda}^k)$ with its scalar part $f \in \dom (\triangle_{p,\max})$. We begin with the limit circle case at $x=0$. Since $\phi \in \dom (d_{1,\max}^t)$ we deduce from the explicit form of $d_1^t$ that $f \in \dom (1/x)_{\max}$, where $1/x$ is the obvious multiplication operator. Since with $p \in (-1/2,1/2)$ $$\textup{id}^{-p}\notin \dom (1/x)_{\max}$$ we deduce $c_2(f)=0$ as in the previous lemma. On the other hand we infer from Proposition \ref{relative-regular} in the degree $k+2$ $$f'(R)-\frac{(k+1-n/2)}{R}f(R)=0.$$ This proves the inclusion $\subset$ in the first statement. Since both sides of the inclusion define self-adjoint extensions and these are maximally symmetric, the inclusion must be an equality.
\\[3mm] For the limit point case at $x=0$ the argumentation is similar up to the fact that the boundary conditions at $x=0$ are redundant by Corollary \ref{ess.s.a}. 
\end{proof}\ \\
\\[-7mm] In contrary to $\dom (\triangle^k_{0,\lambda})$ and $\dom (\triangle^k_{2,\lambda})$, it is not straightforward to determine $\dom (\triangle^k_{\lambda})$ explicitly. However for the purpose of later calculations of zeta determinants it is sufficient to observe that $\triangle^k_{0,\lambda},\triangle^k_{2,\lambda},\triangle^k_{\lambda}$ are Laplacians of the complex \eqref{subcomplex} with relative boundary conditions and hence satisfy the following spectral relation:
$$\textup{Spec}(\triangle^k_{\lambda})\backslash \{0\}=\textup{Spec}(\triangle^k_{0,\lambda})\backslash \{0\}\sqcup \textup{Spec}(\triangle^k_{2,\lambda})\backslash \{0\},$$ where the eigenvalues are counted with their multiplicities.
\\[3mm] Next consider $\mathcal{H}^k(N)$ with the fixed orthonormal basis $\{u^k_i\},i=1,..,\dim \mathcal{H}^k(N)$. Observe that for any $i$ the subspace $C^{\infty}_0((0,R),\langle0\oplus u^k_i,u^k_i\oplus 0\rangle)$ is invariant under $d,d^t$ and we obtain a subcomplex of the de Rham complex
\begin{align*}
0\to C^{\infty}_0((0,R),\langle 0\, \oplus \, &u^k_i,\rangle)\xrightarrow{d} C^{\infty}_0((0,R),\langle u^k_i\oplus 0\rangle) \to 0, \\
&d=(-1)^k\partial_x+\frac{c_k}{x},
\end{align*}
where the action of $d$ is of scalar type under the identification fixed in Remark \ref{scalar-identification}. We continue under this identification. By Proposition \ref{lemma-compatibility} and Theorem \ref{decomposition2} we obtain for the induced self-adjoint extensions
\begin{align}\nonumber
\dom (\triangle_k^{rel})\cap L^2((0,R),&\langle 0\oplus u^k_i\rangle)=\dom (d^t_{\max}d_{\min})=\\ \label{H1-rel}&=\dom\left((-1)^{k+1}\partial_x+\frac{c_k}{x}\right)_{\max}\left((-1)^{k}\partial_x+\frac{c_k}{x}\right)_{\min}, \\ \nonumber
\dom (\triangle_{k+1}^{rel})\cap L^2((0,R),&\langle u^k_i \oplus 0 \rangle)=\dom (d_{\min}d^t_{\max})=\\ \label{H2-rel}&=\dom\left((-1)^{k}\partial_x+\frac{c_k}{x}\right)_{\min}\left((-1)^{k+1}\partial_x+\frac{c_k}{x}\right)_{\max}.
\end{align}
Depending on the explicit value of $c_k=(-1)^k(k-n/2)$ these domains are self-adjoint extensions of regular-singular model Laplacians in limit point or limit circle case at $x=0$. For model Laplacians in the limit circle case at $x=0$ the domains are determined in Subsection \ref{model-operator}. In the limit point case at $x=0$ the boundary conditions at $x=0$ are redundant by Corollary \ref{ess.s.a} and the boundary conditions at $x=R$ are determined in Proposition \ref{relative-regular}.

\section{Functional determinant of the Laplacian with relative Boundary Conditions}\label{even} \
\\[-3mm] We continue in the setup and notation of Section \ref{section-laplace}, where we put $R=1$ for simplicity, and consider the de Rham Laplacian $\triangle_k$ on differential forms of degree $k$. Its self-adjoint extension $\triangle_k^{rel}$ is defined in \eqref{definition} and can be viewed as a self-adjoint operator in $$H^k=L^2([0,1], L^2(\wedge^{k-1}T^*N\oplus \wedge^kT^*N, \textup{vol}(g^N)),dx).$$
We want to identify in each fixed degree $k$ a decomposition 
$$\triangle_k=L_k \oplus \widetilde{\triangle}_k, \quad H^k=H^k_1\oplus H^k_2,$$ compatible with the relative extension $\triangle_k^{rel}$, such that $\widetilde{\triangle}_k$ is the maximal direct sum component, subject to compatibility condition, which is essentially self-adjoint at the cone-singularity in the sense that all its self-adjoint extensions in $H^k_2$ coincide at the cone-tip, in analogy to Definition \ref{coincide}.
\\[3mm] The component $\widetilde{\triangle}_k$ is discussed in [DK]. In this subsection our aim is to understand the structure of $L_k$ and its self-adjoint extension $\mathcal{L}_k^{rel}$ induced in the sense of Definition \ref{induced} by the relative extension $\triangle_k^{rel}$. In particular we want to compute the zeta-regularized determinant of $\mathcal{L}_k^{rel}$ in degrees where it is present.
\\[3mm] Consider the decomposition of $$L^2((0,1),L^2(\wedge^{k-1}T^*N\oplus \wedge^kT^*N)),$$ induced by \eqref{decomposition1}. By Theorem \ref{decomposition2} it is compatible with $\triangle_k^{rel}$. Thus by Lemma \ref{lemma-compatibility} the relative extension $\triangle_k^{rel}$ induces self-adjoint extensions of $\triangle_k$ restricted to each of the orthogonal components of the decomposition. We consider each of the components distinctly.
\begin{prop}\label{prop1}
The relative extension $\triangle_k^{rel}$ induces a self-adjoint extension of $\triangle_k$ restricted to $C^{\infty}_0((0,1),\mathcal{H}^k(N))$. This component contributes to $L_k$ only for $$k \in \left(\frac{m}{2}-2,\frac{m}{2}\right).$$ In these degrees the contribution of the component to the zeta-determinant of $\mathcal{L}_k^{rel}$ is given with $\nu:=|k+1-m/2|$ by $$\left[\frac{\sqrt{2\pi}}{\Gamma(1+\nu)}2^{-\nu}\right]^{\dim \mathcal{H}^k(N)}.$$
\end{prop}
\begin{proof}
Recall from \eqref{H1-rel} in the convention of Remark \ref{scalar-identification}
\begin{align*}
\dom (\triangle_k^{rel})&\cap L^2((0,1),\langle 0\oplus u^k_i\rangle)=\\ &=\dom\left(\left[\partial_x+\frac{(-1)^kc_k}{x}\right]^t_{\max}\left[\partial_x+\frac{(-1)^kc_k}{x}\right]_{\min}\right)=\dom (\triangle^D_{(-1)^kc_k}), 
\end{align*}
where $\{u^k_i\}$ is an orthonormal basis of $\mathcal{H}^k(N)$, $c_k=(-1)^k(k-n/2)$ and $\triangle^D_{(-1)^kc_k}$ denotes the self-adjoint D-extension of $\triangle_{(-1)^kc_k}$, as introduced in Subsection \ref{model-operator}. Note
$$\triangle_{(-1)^kc_k}\equiv \triangle_{k-n/2}=-\frac{d^2}{dx^2}+\frac{1}{x^2}\left[(k+1-m/2)^2-\frac{1}{4}\right]=\triangle_{\nu-1/2},$$
where we put $\nu:=|k+1-m/2|$. We know from Corollary \ref{ess.s.a} that $\triangle_{\nu-1/2}$ is in the limit circle case at $x=0$ and hence not essentially self-adjoint at $x=0$ iff  $$\nu^2-\frac{1}{4}=\left[k+1-\frac{m}{2}\right]^2-\frac{1}{4}<\frac{3}{4}, \ \textup{i.e.} \ k \in \left(\frac{m}{2}-2,\frac{m}{2}\right).$$
Thus we get a contribution to $L_k$ in these degrees only, which is the first part of the statement. 
\\[3mm] Fix such a degree $k\in (m/2-2,m/2)$. Then the contribution to $\mathcal{L}_k^{rel}$ is given by $$\left[\det\nolimits_{\zeta}\triangle^D_{k-n/2}\right]^{\dim \mathcal{H}^k(N)}.$$ Note that for $k\in (m/2-2,m/2)$ we have $(-1)^kc_k=k-n/2\in (-3/2,1/2)$. Depending on the explicit value of $(-1)^kc_k$ we apply Corollaries \ref{ND-extension}, \ref{DN-extension-1/2} and \ref{DN-extension}. We deduce in any case:
$$\dom (\triangle^D_{k-n/2})=\{f \in \dom (\triangle_{\nu-1/2, \max})|c_2(f)=0, f(1)=0\},$$
where $c_2(f)$ refers to the asymptotics in Proposition \ref{laplacian-maximal} or equivalently in \eqref{laplacian-maxB}. 
\\[3mm] We deduce the explicit value of $\det_{\zeta}\triangle^D_{k-n/2}$ from Corollary \ref{4}. 
\end{proof}

\begin{prop}\label{prop2}
The relative extension $\triangle_k^{rel}$ induces a self-adjoint extension of $\triangle_k$ restricted to $C^{\infty}_0((0,1),\mathcal{H}^{k-1}(N))$. This component contributes to $L_k$ only for $$k \in \left(\frac{m}{2},\frac{m}{2}+2\right).$$ In these degrees the contribution of the component to the zeta-determinant of $\mathcal{L}_k^{rel}$ is given by 
\begin{align*}\begin{array}{rl}
\left[\sqrt{\pi /2}\right]^{\dim \mathcal{H}^{k-1}(N)}, \ & \textup{if} \ \dim M=m \ \textup{is even}, k=m/2+1, \\
\left[2\right]^{\dim \mathcal{H}^{k-1}(N)}, \ & \textup{if} \ \dim M=m \ \textup{is odd}, \ k=n/2+1, \\
\left[2/3\right]^{\dim \mathcal{H}^{k-1}(N)}, \ & \textup{if} \ \dim M=m \ \textup{is odd}, \ k=n/2+2.\end{array}
\end{align*}
\end{prop}
\begin{proof}
Recall from \eqref{H2-rel} in the convention of Remark \ref{scalar-identification}
\begin{align*}
\dom (\triangle_k^{rel})&\cap L^2((0,1),\langle u^{k-1}_i \oplus 0\rangle)=\\ &=\dom\left(\left[\partial_x+\frac{(-1)^kc_{k-1}}{x}\right]^t_{\min}\left[\partial_x+\frac{(-1)^kc_{k-1}}{x}\right]_{\max}\right)=\dom (\triangle^N_{(-1)^kc_{k-1}}), 
\end{align*}
where $\{u^{k-1}_i\}$ is an orthonormal basis of $\mathcal{H}^{k-1}(N)$, $c_{k-1}=(-1)^{k-1}(k-1-n/2)$ and $\triangle^N_{(-1)^kc_{k-1}}$ denotes the self-adjoint N-extension of $\triangle_{(-1)^kc_{k-1}}$, as introduced in Subsection \ref{model-operator}. Note
$$\triangle_{(-1)^kc_{k-1}}\equiv \triangle_{n/2+1-k}=-\frac{d^2}{dx^2}+\frac{1}{x^2}\left[(k-1-m/2)^2-\frac{1}{4}\right]=\triangle_{\nu-1/2},$$
where we put $\nu:=|k-1-m/2|$. We know from Corollary \ref{ess.s.a} that $\triangle_{\nu-1/2}$ is in the limit circle case at $x=0$ and hence not essentially self-adjoint at $x=0$ iff  $$\nu^2-\frac{1}{4}=\left[k-1-\frac{m}{2}\right]^2-\frac{1}{4}<\frac{3}{4}, \ \textup{i.e.} \ k \in \left(\frac{m}{2},\frac{m}{2}+2\right).$$
Thus we get a contribution to $L_k$ in these degrees only, which is the first part of the statement. 
\\[3mm] Fix such a degree $k\in (m/2,m/2+1)$. Then the contribution to $\mathcal{L}_k^{rel}$ is given by $$\left[\det\nolimits_{\zeta}\triangle^N_{n/2+1-k}\right]^{\dim \mathcal{H}^{k-1}(N)}.$$
Unfortunately the explicit form of the domain $\dom (\triangle^N_{n/2+1-k})$ cannot be presented in such a homogeneous way as in the previous proposition. So we need to discuss different cases separately.
\\[3mm] If $\dim M=m$ is even, then the only degree $k \in (m/2, m/2+2)$ is $k=m/2+1$. Then $n/2+1-k=-1/2, \nu=0$ and we obtain with Corollary \ref{DN-extension-1/2} $$\dom (\triangle_{n/2+1-k}^N)=\{f \in \dom (\triangle_{\nu-1/2,\max})|c_2(f)=0, f'(1)-\frac{1}{2}f(1)=0\},$$
where $c_2(f)$ refers to the asymptotics \eqref{laplacian-maxA}. The contribution to the zeta-determinant of $\mathcal{L}_k^{rel}$ follows now from Corollary \ref{2}. 
\\[3mm] If $\dim M=m$ is odd, then the only degrees $k \in (m/2, m/2+2)$ are $k=n/2+1, n/2+2$. 
\\[3mm] For $k=n/2+1$ we have $n/2+1-k=0, \nu=1/2$ and we obtain from Corollary \ref{ND-extension} $$\dom (\triangle_{n/2+1-k}^N)=\{f \in \dom (\triangle_{\nu-1/2,\max})|c_1(f)=0, f'(1)=0\},$$
where $c_1(f)$ refers to the asymptotics \eqref{laplacian-maxB}. The contribution to the zeta-determinant of $\mathcal{L}_k^{rel}$ in this case follows from Corollary \ref{2'}.
\\[3mm] For the second case $k=n/2+2$ we have $n/2+1-k=-1, \nu=1/2$. So we obtain from Corollary \ref{DN-extension} $$\dom (\triangle_{n/2+1-k}^N)=\{f \in \dom (\triangle_{\nu-1/2,\max})|c_2(f)=0, f'(1)-f(1)=0\},$$
where $c_2(f)$ refers to the asymptotics \eqref{laplacian-maxB}. The contribution to the zeta-determinant of $\mathcal{L}_k^{rel}$ in this case follows from Corollary \ref{1'}. \\[3mm] Now all the possible cases are discussed and the proof is complete.
\end{proof}

\begin{prop}\label{prop3}
The relative extension $\triangle_k^{rel}$ induces a self-adjoint extension of $\triangle_k$ restricted to $C^{\infty}_0((0,1),\{0\}\oplus E_{\lambda}^k)$ for $\lambda \in V_k=\textup{Spec}\triangle_{k,ccl,N}\backslash \{0\}$. This component contributes to $L_k$ only for $$k \in \left(\frac{m}{2}-2,\frac{m}{2}\right), \ \la < 1-\left[k+\frac{1}{2}-\frac{n}{2}\right]^2.$$ In this case the contribution of the component to the zeta-determinant of $\mathcal{L}_k^{rel}$ is given by $$\frac{\sqrt{2\pi}}{\Gamma(1+\nu)2^{\nu}}, \ \textup{where} \ \nu:=\sqrt{\la +\left[k+\frac{1}{2}-\frac{n}{2}\right]^2}.$$
\end{prop}
\begin{proof}
We infer from \eqref{lapl1} that $\triangle_k$ acts on $C^{\infty}_0((0,1),\{0\}\oplus E_{\lambda}^k)$ with $\lambda \in V_k$ as a rank-one model Laplacian
$$-\frac{d^2}{dx^2}+\frac{1}{x^2}\left(\la + \left[k+\frac{1}{2}-\frac{n}{2}\right]^2-\frac{1}{4}\right),$$ under the identification of elements in $C^{\infty}_0((0,1),\{0\}\oplus E_{\lambda}^k)$ with their scalar parts in $C^{\infty}_0(0,1)$ in the convention of Remark \ref{scalar-identification}. This operator is in the limit circle case at $x=0$ and hence not essentially self-adjoint at $x=0$ iff 
\begin{align*}
\la +\left[k+\frac{1}{2}-\frac{n}{2}\right]^2-\frac{1}{4}<\frac{3}{4},
\ \textup{i.e. } \ k \in \left(\frac{m}{2}-2,\frac{m}{2}\right), \ \la < 1-\left[k+\frac{1}{2}-\frac{n}{2}\right]^2.
\end{align*}
This proves the first part of the statement. Fix such $k$ and $\lambda$. Observe now by Lemma \ref{relative0} $$\dom(\triangle_k^{rel})\cap L^2((0,1),\{0\}\oplus E_{\la}^k)=\{f \in \dom (\triangle_{p,\max}) | c_2(f)=0, f(1)=0\}.$$
The result now follows from Corollary \ref{4}.
\end{proof}
\begin{prop}\label{prop4}
The relative extension $\triangle_k^{rel}$ induces a self-adjoint extension of $\triangle_k$ restricted to $C^{\infty}_0((0,1),d_N E_{\lambda}^{k-2})\oplus \{0\}$ for $\lambda \in V_{k-2}=\textup{Spec}\triangle_{k-2,ccl,N}\backslash \{0\}$. This component contributes to $L_k$ only for $$k \in \left(\frac{m}{2},\frac{m}{2}+2\right), \ \la < 1-\left[k-\frac{3}{2}-\frac{n}{2}\right]^2.$$ In this case the contribution of the component to the zeta-determinant of $\mathcal{L}_k^{rel}$ is given by $$\sqrt{2\pi}\frac{\nu +m/2+1-k}{\Gamma(1+\nu)\cdot 2^{\nu}}, \ \textup{where} \ \nu:=\sqrt{\la +\left[k-\frac{3}{2}-\frac{n}{2}\right]^2}.$$
\end{prop}
\begin{proof}
We infer from \eqref{lapl1} that $\triangle_k$ acts on $C^{\infty}_0((0,1),d_N E_{\lambda}^{k-2})$ with $\lambda \in V_{k-2}$ as a rank-one model Laplacian
$$-\frac{d^2}{dx^2}+\frac{1}{x^2}\left(\la + \left[k-\frac{3}{2}-\frac{n}{2}\right]^2-\frac{1}{4}\right),$$ under the identification of elements with their scalar parts as before. This operator is in the limit circle case at $x=0$ and hence not essentially self-adjoint at $x=0$ iff 
\begin{align*}
\la +\left[k-\frac{3}{2}-\frac{n}{2}\right]^2-\frac{1}{4}<\frac{3}{4},
\ \textup{i.e. } \ k \in \left(\frac{m}{2},\frac{m}{2}+2\right), \ \la < 1-\left[k-\frac{3}{2}-\frac{n}{2}\right]^2.
\end{align*}
This proves the first part of the statement. Fix such $k$ and $\lambda$. Observe now by Lemma \ref{relative2} 
\begin{align*}
\dom(\triangle_k^{rel})\cap L^2((0,1),\{0\}\oplus E_{\la}^k)=\{f \in \dom (\triangle_{p,\max}) | \\c_2(f)=0, f'(1)-(k-1-n/2)f(1)=0\}.
\end{align*}
The result now follows from Corollary \ref{1}.
\end{proof}

\begin{prop}\label{prop5}
The relative extension $\triangle_k^{rel}$ induces a self-adjoint extension of $\triangle_k$ restricted to $C^{\infty}_0((0,1),\widetilde{E}_{\la}^{k-1})$ with $\la \in V_{k-1}=\textup{Spec}\triangle_{k-1,ccl,N}\backslash \{0\}$. This component contributes to $L_k$ only for $$k \in \left(\frac{m}{2}-2,\frac{m}{2}+2\right), \ \la < 4-\left[k-\frac{m}{2}\right]^2.$$ The contribution of the component to the zeta-determinant of $\mathcal{L}_k^{rel}$ is given by $$2\pi \frac{(\nu -k+m/2)}{\Gamma(1+\nu)^2}2^{-2\nu}, \textup{where} \ \nu:=\sqrt{\la +\left[k-\frac{m}{2}\right]^2}.$$
\end{prop}
\begin{proof}
The space $\widetilde{E}_{\la}^{k-1}, \la \in V_{k-1}$ is an orthogonal sum of $S_0$-eigenspaces to eigenvalues (see [BL2, Section 2]) $$p_{\pm}^k(\la):=\frac{(-1)^k}{2}\pm \sqrt{\left(k-\frac{m}{2}\right)^2+\lambda}.$$ Put $$a_{\pm}^k(\la):=p_{\pm}^k(\la) \cdot (p_{\pm}^k(\la)+(-1)^k).$$ The restriction of $\triangle_k$ to $C^{\infty}_0((0,1),\widetilde{E}_{\la}^{k-1})$ decomposes into 
$$\left(-\frac{d^2}{dx^2}+\frac{a_+^k(\la)}{x^2}\right)\oplus \left(-\frac{d^2}{dx^2}+\frac{a_-^k(\la)}{x^2}\right)$$ in correspondence to the decomposition of $\widetilde{E}_{\la}^{k-1}$ into the $S_0$-eigenspaces. This decomposition is not compatible with the relative boundary conditions, which is clear from the relative boundary conditions at the cone base. Nevertheless we infer from the decomposition, that the restriction of $\triangle_k$ to $C^{\infty}_0((0,1),\widetilde{E}_{\la}^{k-1})$ contributes to $L_k$ only if $$\left|k-\frac{m}{2}\right|<2, \quad \la \leq 4-\left(k-\frac{m}{2}\right)^2,$$ since for the complementary case both $a_{+}^k(\la)$ and $a_{-}^k(\la)$ are $\geq 3/4$. This proves the first part of the statement. In order to compute the contribution of the component to the determinant of $\mathcal{L}_k^{rel}$, we study as in \eqref{subcomplex} the associated de Rham complex:
\begin{align*}
0\!\to\! C^{\infty}_0((0,1), \{0\}\!\oplus\! E_{\la}^{k-1})\!\xrightarrow{d_0}\!C^{\infty}_0((0,1), \widetilde{E}_{\la}^{k-1})\!\xrightarrow{d_1}\!C^{\infty}_0((0,1), d_N E_{\la}^{k-1}\!\oplus\! \{0\})\!\to\!  0.
\end{align*}
Note that $d_0^td_0$ and $d_1d_1^t$ both act as rank-one model Laplacians
$$-\frac{d^2}{dx^2}+\frac{1}{x^2}\left(\la +\left[k-\frac{1}{2}-\frac{n}{2}\right]^2-\frac{1}{4}\right),$$
under the identification of elements with their scalar parts, as before. The relative boundary conditions turn the complex into a Hilbert complex with the corresponding self-adjoint extensions of $d_0^td_0$ and $d_1d_1^t$ denoted by $\triangle_0$, $\triangle_1$ respectively. The contribution to the determinant of $\mathcal{L}_k^{rel}$ is then given by $$\det\nolimits_{\zeta}\triangle_0 \cdot \det\nolimits_{\zeta}\triangle_1.$$
We obtain with $\nu:=\sqrt{\la +(k-m/2)^2}$
\begin{align*}
\dom (\triangle_0)=\{f \in \dom (\triangle_{\nu-1/2,\max})| f(x)=O(\sqrt{x}),x\to 0; f(1)=0\}, \\
\dom (\triangle_1)=\{f \in \dom (\triangle_{\nu-1/2,\max})| f(x)=O(\sqrt{x}),x\to 0, \\ f'(1)-(k-n/2)f(1)=0\}.
\end{align*}
For $\nu\in (0,1)$ these domains were determined in Lemma \ref{relative0} and Lemma \ref{relative2}, with the asymptotics $f(x)=O(\sqrt{x}),x\to 0$ being expressed by $c_2(f)=0$. The coefficient $c_2(f)$ refers to the relation \eqref{laplacian-maxB}.
\\[3mm] For $\nu\geq 1$ the operator $\triangle_{\nu -1/2}$ is in the limit point case, see Corollary \ref{ess.s.a}. So the condition on the asymptotic behaviour at $x=0$ is redundant and the domains are determined only by the relative boundary conditions at $x=1$, which are computed in Proposition \ref{relative-regular}.
\\[3mm] Applying Corollaries \ref{1} and \ref{4} we obtain 
\begin{align*}
\det\nolimits_{\zeta}\triangle_{0}=\frac{\sqrt{2\pi}}{\Gamma (1+\nu)2^{\nu}},\\
\det\nolimits_{\zeta}\triangle_{1}=\sqrt{2\pi}\frac{\nu-k+m/2}{\Gamma (1+\nu)2^{\nu}}.
\end{align*}
Multiplication of both expressions gives the result.
\end{proof} \ \\
\\[-7mm] Before we write down an explicit expression for $\det_{\zeta}\mathcal{L}_k^{rel}$, let us introduce some simplifying notation. Put:
\begin{align}
&A_k:=\{\nu=\sqrt{\lambda +[k+1-m/2]^2}|\lambda \in \textup{Spec}\triangle_{k,ccl,N}, \nonumber \\ &\hspace{30mm}0\leq \lambda < 1-[k+1-m/2]^2\}, \label{A_k}\\
&\widetilde{A}_k:=\{\nu=\sqrt{\lambda +[k+1-m/2]^2}|\lambda \in \textup{Spec}\triangle_{k,ccl,N}\backslash \{0\}, \nonumber \\ &\hspace{30mm}0< \lambda < 1-[k+1-m/2]^2\}, \label{A-tilde_k}\\
&B_k:=\{\nu=\sqrt{\lambda +[k-m/2]^2}|\lambda \in \textup{Spec}\triangle_{k-1,ccl,N}\backslash \{0\}, \nonumber \\ &\hspace{30mm}0< \lambda < 4-[k-m/2]^2\}. \label{B_k}
\end{align}
Moreover we write 
\begin{align}
P_k:=\left\{\begin{array}{rl} (\sqrt{\pi /2})^{\dim \mathcal{H}^{k-1}(N)}, & \textup{for $k=m/2+1$ if $\dim M=m$ even}, \nonumber \\
2^{\dim \mathcal{H}^{k-1}(N)}, & \textup{for $k=n/2+1$ if $\dim M=n+1$ odd}, \nonumber \\
(2/3)^{\dim \mathcal{H}^{k-1}(N)}, & \textup{for $k=n/2+2$ if $\dim M=n+1$ odd}.
\end{array}\right. \nonumber \\
\label{P_k-loya}
\end{align}
The preceeding computations imply that $L_k$ is a finite direct sum of model Laplace operators, a regular-singular Sturm-Liouville operator with matrix coefficients, and in fact does not occur for $|k-m/2|\geq 2$. This corresponds to the general fact, see [BL2, Theorem 3.7, Theorem 3.8] that the Laplace operator on $k-$forms over a compact manifold with an isolated singularity is "essentially self-adjoint" at the cone tip outside of the middle degrees, i.e. for $|k-m/2|\geq 2$. 
\\[3mm] Therefore the complete determinant of $\mathcal{L}_k^{rel}$ is given simply by a product of finitely many contributions, determined in Propositions \ref{prop1} $-$ \ref{prop5}, depending on the choice of a degree. This proves the central result of this subsection.
\begin{thm}\label{total-finite-determinant}
The self-adjoint operator $\mathcal{L}_k^{rel}$ is non-trivial only for degrees $$k\in \left(\frac{m}{2}-2,\frac{m}{2}+2\right).$$
In these degrees the zeta-determinant of $\mathcal{L}_k^{rel}$ is given as follows, where we use the notation established in \eqref{A_k} $-$ \eqref{P_k-loya}:
\begin{enumerate}
\item For $k\in (m/2-2,m/2)$ we have
\begin{align*}
\det\nolimits_{\zeta}\mathcal{L}_k^{rel}=\prod\limits_{\nu \in A_k}\frac{\sqrt{2\pi}}{2^{\nu}\Gamma(1+\nu)}\prod\limits_{\nu \in B_k}2\pi \frac{\nu-k+m/2}{2^{2\nu}\Gamma(1+\nu)^2}.
\end{align*}
\item For $k\in (m/2, m/2+2)$ we have
\begin{align*}
\det\nolimits_{\zeta}\mathcal{L}_k^{rel}=\prod\limits_{\nu \in \widetilde{ A}_{k-2}}\frac{\sqrt{2\pi}(\nu+m/2+1-k)}{2^{\nu}\Gamma(1+\nu)}\prod\limits_{\nu \in B_k}2\pi \frac{\nu-k+m/2}{2^{2\nu}\Gamma(1+\nu)^2}\cdot P_k.
\end{align*}
\item For $\dim M=m$ even and $k=m/2$ we have
\begin{align*}
\det\nolimits_{\zeta}\mathcal{L}_k^{rel}=\prod\limits_{\nu \in B_k}2\pi \frac{\nu}{2^{2\nu}\Gamma(1+\nu)^2}.
\end{align*} 
\end{enumerate}
\end{thm} 

\section{References}\
\\[-1mm] [AS] Editors M. Abramowitz, I.A. Stegun \emph{"Handbook of math. Functions"} AMS. 55.
\\[3mm] [BGKE] M. Bordag, B. Geyer, K. Kirsten, E. Elizalde \emph{"Zeta function determinant of the Laplace Operator on the D-dimensional ball"}, Comm. Math. Phys. 179, no. 1, 215-234, (1996)
\\[3mm] [BKD] M. Bordag, K. Kirsten, J.S. Dowker \emph{"Heat-kernels and functional determinants on the generalized cone"}, Comm. Math. Phys. 182, 371-394, (1996)
\\[3mm] [BL1] J. Br\"{u}ning, M. Lesch \emph{"Hilbert complexes"}, J. Funct. Anal. 108, 88-132 (1992)
\\[3mm] [BL2] J. Br\"{u}ning, M. Lesch \emph{"K\"{a}hler-Hodge Theory for conformal complex cones"}, Geom. Funct. Anal. 3, 439-473 (1993)
\\[3mm] [Br] J. Br\"{u}ning \emph{$L^2$-index theorems for certain complete manifolds}, J. Diff. Geom. 32 491-532 (1990)
\\[3mm] [BS] J. Br\"{u}ning, R. Seeley \emph{"An index theorem for first order regular singular operators"}, Amer. J. Math 110, 659-714, (1988)
\\[3mm] [BS2] J. Br\"{u}ning, R. Seeley \emph{"Regular Singular Asymptotics"}, Adv. Math. 58, 133-148 (1985)
\\[3mm] [BS3] J. Br\"{u}ning, R. Seeley \emph{"The resolvent expansion for second order Regular Singular Operators"}, J. Funct. Anal. 73, 369-429 (1987)
\\[3mm] [BV2] B. Vertman \emph{"Analytic Torsion of a bounded generalized cone"}, preprint, arXiv:0808.0449 (2008)
\\[3mm] [C] C. Callias \emph{"The resolvent and the heat kernel for some singular boundary problems"}, Comm. Part. Diff. Eq. 13, no. 9, 1113-1155 (1988)
\\[3mm] [Ch1] J. Cheeger \emph{On the spectral geometry of spaces with conical singularities}, Proc. Nat. Acad. Sci. 76, 2103-2106 (1979)
\\[3mm] [Ch2] J. Cheeger \emph{"Spectral Geometry of singular Riemannian spaces"}, J. Diff. Geom. 18, 575-657 (1983)
\\[3mm] [DK] J.S. Dowker and K. Kirsten \emph{"Spinors and forms on the ball and the generalized cone"}, Comm. Anal. Geom. Volume 7, Number 3, 641-679, (1999)
\\[3mm] [GRA] I.S. Gradsteyn, I.M Ryzhik, Alan Jeffrey \emph{"Table of integrals, Series and Products"}, 5th edition, Academic Press, Inc. (1994)
\\[3mm] [KLP1] K. Kirsten, P. Loya, J. Park \emph{"Functional determinants for general self-adjoint extensions of Laplace-type operators resulting from the generalized cone"}, Manuscripta Mathematica 125, 95-126, (2008)
\\[3mm] [KLP2] K. Kirsten, P. Loya, J. Park, with an Appendix by B. Vertman \emph{"Exotic expansions and pathological properties of zeta-functions on conic manifolds"}, Journal of Geometric Analysis 18, 835-888, (2008) 
\\[3mm] [L] M. Lesch \emph{"Determinants of regular singular Sturm-Liouville operators"}, Math. Nachr. (1995)
\\[3mm] [L1] M. Lesch \emph{"Operators of Fuchs Type, Conical singularities, and Asymptotic Methods"}, Teubner, Band 136 (1997)
\\[3mm] [L3] M. Lesch \emph{"The analytic torsion of the model cone"}, Columbus University (1994), unpublished notes.
\\[3mm] [LMP] P. Loya, P. McDonald, J. Park \emph{Zeta regularized determinants for conic manifolds}, J. Func. Anal. (2006)
\\[3mm] [M] E. Mooers \emph{"Heat kernel asymptotics on manifolds with conic singularities"} J. Anal. Math. 78, 1-36 (1999) 
\\[3mm] [O] F.W. Olver \emph{"Asymptotics and special functions"} AKP Classics (1987)
\\[3mm] [P] L. Paquet \emph{"Probl'emes mixtes pour le syst'eme de Maxwell"}, Annales Facultè des Sciences Toulouse, Volume IV, 103-141 (1982)
\\[3mm] [ReS] M. Reed, B. Simon \emph{"Methods of Mathematical Physics"}, Vol. II, Acad.N.J. (1979)
\\[3mm] [Ru] W. Rudin \emph{"Functional Analysis"}, Second Edition, Mc. Graw-Hill, Inc. Intern. Series in pure and appl. math. (1991)
\\[3mm] [W] J. Weidmann \emph{"Spectral theory of ordinary differential equations"}, Springer-Verlag Berlin-Heidelberg, Lecture Notes in Math. 1258, (1987)
\\[3mm] [W2] J. Weidmann \emph{"Linear Operators in Hilbert spaces"}, Springer-Verlag, New York, (1980)
\\[3mm] [WT] G.N. Watson \emph{"A treatise on the theory of Bessel functions"}, Camb. Univ. Press (1922)
\\[3mm] [WW] E. T. Whittaker, G.N. Watson \emph{"A course on modern analysis"}, Camb. Univ. Press (1946)

\end{document}